\documentclass{article}
\usepackage[margin=1.2in]{geometry}
\usepackage{amsmath,amssymb}
\usepackage{braket}
\usepackage{xcolor}
\usepackage{graphicx,subcaption,mathtools}
\usepackage{float} 
\usepackage[
backend=bibtex,  
style=numeric   
]{biblatex}
\newcommand{\Tr}{\operatorname{Tr}}

\DeclareMathOperator*{\esssup}{esssup}
\DeclareMathOperator*{\argmax}{argmax}

\let\rm\mathrm
\usepackage{amsfonts}
\usepackage{dsfont}

\usepackage{amsthm}
\date{}
\theoremstyle{plain} 
\newtheorem{theorem}{Theorem}

\newtheorem{proposition}{Proposition}
\newtheorem{corollary}{Corollary}

\theoremstyle{definition}

\theoremstyle{remark} 

\title{Universal Sequential Changepoint Detection of Quantum Observables via Classical Shadows}
\author{Matteo Zecchin$^\dag$, Osvaldo Simeone$^\ddagger$ , and Aaditya Ramdas$^\mathsection$}
\author{%
  Matteo Zecchin\thanks{Communication Systems Department, EURECOM, 06904 Sophia Antipolis, France.
    \texttt{zecchin@eurecom.fr} },
    \quad
  Osvaldo Simeone\thanks{Institute for Intelligent Networked Systems, Northeastern University London, London E1 8PH, UK.
    \texttt{o.simeone@northeastern.edu},}
  \quad
   and Aaditya Ramdas\thanks{ Department of Statistics \& Data Science, Machine Learning Department, Carnegie Mellon University, Pittsburgh, USA;
    \texttt{aramdas@cmu.edu}}%
}
\begin{document}

\maketitle

\begin{abstract}
We study sequential quantum changepoint detection in settings where the pre- and post-change regimes are specified through constraints on the expectation values of a finite set of observables. We consider an architecture with separate measurement and detection modules, and assume that the observables relevant to the detector are unknown to the measurement device. For this scenario, we introduce shadow-based sequential changepoint e-detection (eSCD), a novel protocol that combines a universal measurement strategy based on classical shadows with a nonparametric sequential test built on e-detectors. Classical shadows provide universality with respect to the detector's choice of observables, while the e-detector framework enables explicit control of the average run length (ARL) to false alarm. Under an ARL constraint, we establish finite-sample guarantees on the worst-case expected detection delay of eSCD. Numerical experiments validate the theory and demonstrate that eSCD achieves performance competitive with observable-specific measurement strategies, while retaining full measurement universality.

\end{abstract}
\section{Introduction}
\subsection{Motivation}
Sequential changepoint detection aims to detect, as quickly as possible, a change in the statistical behavior of a stream of observations, while controlling the frequency of false alarms \cite{page1954continuous}. When the observations correspond to measurement outcomes on quantum systems, the task is referred to as sequential \emph{quantum} changepoint detection  \cite{sentis2018online,fanizza2023ultimate}.

In its standard formulation, sequential quantum changepoint detection assumes a changepoint model in which the quantum state is in a {known} pre-change state up to an unknown changepoint time, and in a {known} post-change state thereafter \cite{akimoto2011discrimination,sentis2016quantum,sentis2018online,fanizza2023ultimate,grootveld2026asymptotically}. In this setting, the problem reduces to designing measurements that best separate the two hypotheses \cite{hayashi2001asymptotics}, and then applying a classical sequential test to the induced measurement outcome process \cite{page1954continuous}. In this regime, performance limits and optimal designs are characterized in terms of information quantities that capture how distinguishable the given  pre- and post-change states become after measurement \cite{fanizza2023ultimate}. Recent extensions replace the simple post-change model with a {composite}  alternative, in which the post-change state belongs to a family of states \cite{grootveld2026asymptotically}.

\begin{figure}
	\centering
	\includegraphics[width=0.75\textwidth]{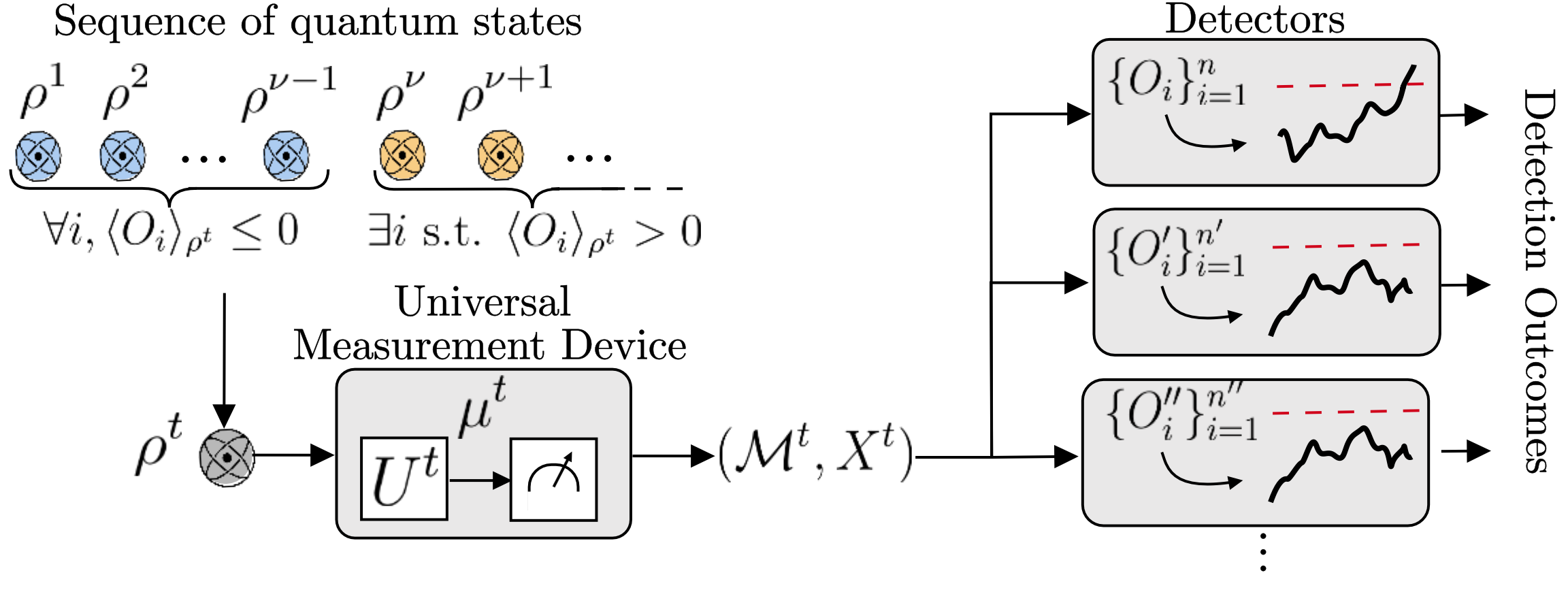}
	\caption{ A sequence of quantum states  $\{\rho^t\}_{t\geq 1}$ is processed by a measurement device that selects a measurement setting $\mathcal{M}^t$ and obtains a measurement output $X^t$ at each time step $t$. The pair $(\mathcal{M}^t,X^t)$ is then processed by a detection algorithm that aims to determine whether the expected value $\langle O_i\rangle_{\rho^t}$ of at least one observable $\{O_i\}^n_{i=1}$ has become positive.  The measurement device must be designed to support different detection algorithms tailored to different sets of observables  $\{O_i\}^n_{i=1}$.}
	\label{fig:setting}
\end{figure}

In contrast to this line of research, in many quantum applications, the experimenter is practically interested in characterizing a device or quantum state through a small number of \emph{observables}, rather than via a full specification of the system's density matrix \cite{aaronson2018shadow,huang2020predicting}. For example, in variational quantum algorithms, such as the variational quantum eigensolver (VQE), the primary goal is to estimate the expectation value of a problem Hamiltonian \cite{cerezo2021variational,simeone2022introduction};  and in
quantum simulation of condensed matter systems or molecular dynamics, researchers measure local observables, such as spin-spin correlation functions, which capture the relevant physics without requiring full tomographic reconstruction \cite{feynman2018simulating,georgescu2014quantum}. Other examples include quantum error correction \cite{gottesman1997stabilizer} and quantum sensing \cite{degen2017quantum}. 

Furthermore, it is common in quantum information theory and quantum cryptography to adopt a framework in which the quantum \emph{measurement devices} (in the ``lab'') are treated as separate entities from the classical \emph{detection devices}   performing statistical testing and data analysis (the ``statistician" or ``verifier") \cite{hensen2015loophole}. This conceptual, and often physical, separation is fundamental, for instance, to device-independent protocols, where security or validity guarantees depend only on observable statistics rather than on detailed models of the quantum devices themselves \cite{arnon2018practical}.

Motivated by these  considerations, we consider the setting illustrated in Figure~\ref{fig:setting}, which differs from prior art in two fundamental ways:\begin{itemize}
    \item \textbf{Observable-based changepoint:} The pre- and post-change regimes are determined by constraints on the expectation values of a finite collection of observables. Accordingly,  for any fixed set of observables, the induced pre- and post-change families of quantum states are both composite, and the resulting measurement statistics are nonparametric.
    \item \textbf{Universal measurement device:} The set of observables to be monitored may be chosen by a detection device only after data acquisition by a separate measurement device. Since the measurement stage is agnostic to the downstream choice of observables, the measurement strategy cannot be tailored to a specific change model. Instead, the procedure must be \emph{universal} supporting arbitrary observable sets (within a given class). 
\end{itemize} 

To address these challenges, we propose a novel protocol, referred to as \emph{shadow-based sequential changepoint e-detection}  (eSCD), that combines a universal measurement strategy
based on classical shadows with a  nonparametric sequential change detector. Through this design, eSCD leverages the ``measure once, test many'' property of classical shadows \cite{huang2020predicting} to define a universal measurement strategy whose outcomes can be used to detect changes via the recently developed framework of e-detectors \cite{shin2022detectors}. Furthermore, the e-detector framework enables explicit
control of the average run length (ARL) to false alarm, while supporting finite-sample guarantees on the worst-case expected
detection delay.

\subsection{Related Work}

Quantum changepoint detection was originally studied in the \emph{offline}  setting, in which the goal is localization of a changepoint within a finite sequence of quantum states \cite{akimoto2011discrimination,sentis2017exact,sentis2016quantum}. Later extensions focused on  \emph{sequential} settings in which the pre-change and post-change states are \emph{known} pure states \cite{sentis2018online}. 

In sequential scenarios, fundamental limits for quickest quantum changepoint detection were shown to be characterized by the measured quantum relative entropy \cite{simeone2026classical} between pre- and post-change state \cite{fanizza2023ultimate}. Schemes that asymptotically achieve these bounds were obtained by applying classical (non-quantum) changepoint detection algorithms, such as CUSUM \cite{page1954continuous}, to the sequence of measurement outcomes generated by the optimal measurement attaining the measured quantum relative entropy. 

For changepoint models with composite post-change states, the performance of sequential changepoint detection has been studied under the assumption that the post-change family admits a suitable discretization \cite{fanizza2023ultimate} or that the post-change state can be estimated \cite{grootveld2026asymptotically}. In this regime, fundamental limits for sequential changepoint detection were characterized when joint measurements over growing blocks of subsequent quantum states are allowed. Recently, windowed CUSUM-type procedures have been shown to achieve these limits asymptotically when the block size grows to infinity \cite{grootveld2026asymptotically}. Other recent extensions of quantum changepoint detection include continuous families of quantum states \cite{john2025fundamental} and changepoint detection for quantum channels \cite{nakahira2025unambiguous}.

In the setting considered in this paper, classical sequential changepoint schemes that rely on optimizing the measurement strategy are not directly applicable. In fact, as shown in Figure \ref{fig:setting}, we impose that the measurement device be universal, i.e., agnostic to the set of tested observables.  The \emph{classical shadow} framework provides randomized measurement protocols and estimation procedures that enable the prediction of many properties of quantum states without tailoring the measurement to any specific downstream quantity \cite{huang2020predicting}. In our setting, we use this ``measure once, test many'' property to construct, for each candidate observable, a sequential test statistic.

The induced pre- and post-change measurement statistics returned by classical shadows are composite and \emph{nonparametric}. In fact, the  pre- and post-change measurement statistics are only specified by the mean of the given observables.  To address the resulting testing problem, we adopt the recently developed framework of \emph{e-detectors} \cite{shin2022detectors},  a nonparametric extension of standard techniques such as CUSUM. E-detectors build sequential changepoint tests from e-processes \cite{ramdas2025hypothesis}, offering finite-sample control of false alarms under minimal distributional assumptions and enjoying strong efficiency guarantees \cite{ram2026asymptotically}. 

A related methodology based on e-processes has recently been developed in quantum settings for sequential composite hypothesis testing \cite{zecchin2025quantum} and anytime quantum tomography \cite{cumitini2026anytime}.

\subsection{Contributions}
In this work, as illustrated in Figure \ref{fig:setting}, we introduce eSCD, a novel quantum changepoint detection protocol incorporating a universal measurement device based on classical shadows and a testing device based on e-detectors. 
Our main contributions are as follows:
\begin{itemize}
	\item \textbf{Observable-based changepoint detection:} We formulate a novel sequential quantum changepoint detection problem in which the changepoint model is specified via constraints on the expectations of a collection of observables. We consider an architecture with
separate measurement and detection modules, and assume that the observables relevant to the detector are unknown to the measurement device. 
	\item \textbf{Shadow-based sequential changepoint e-detection (eSCD):} We propose eSCD, a novel protocol that combines a universal detector based on classical shadows \cite{huang2020predicting} with e-detectors \cite{shin2022detectors} for nonparametric changepoint detection. 
	\item \textbf{Theoretical guarantees:} We prove that eSCD controls the \emph{average run length} (ARL) to false alarm, and establish finite-sample bounds on the worst-case expected detection delay \cite{shin2022detectors}.
	\item \textbf{Experimental validation:} We present numerical experiments benchmarking eSCD against detection schemes based on matched measurements.
\end{itemize}

The rest of the paper is organized as follows. Section~\ref{sec:prof_definition} formulates the sequential quantum changepoint detection problem based on observables, and introduces the key performance metrics. Section~\ref{sec:escd} presents the proposed eSCD framework, including the observable-agnostic randomized measurement policy based on classical shadows and the e-detectors, together with the resulting theoretical guarantees. In Section \ref{sec:experiments} we provide experimental results evaluating ARL control and efficiency of eSCD, including comparisons with changepoint detection schemes based on matched measurement baselines. Section \ref{sec:conclusions} concludes the paper with future research directions.

\section{Problem Definition}
\label{sec:prof_definition}
In this section, we formulate the problem under study of sequential quantum changepoint detection of quantum observables and define the relevant performance metrics.

\subsection{Sequential Changepoint Detection}
\label{sec:gen_prob}
As illustrated in Figure \ref{fig:setting}, we target a setting with separate measurement and detection devices, in which measurement may be used to detect different types of changes in the stream of quantum states under observation. We first describe the changepoint model assumed in this work

\hspace{\parindent} \emph{1) Changepoint Model:} As seen in Figure \ref{fig:setting}, the quantum measurement device has access to a sequence of $d$-qubit systems prepared in quantum states $\{\rho^t\}_{t \ge 1}$ that are indexed by a discrete time step $t \in \{1,2,\dots\}$.  Each state $\rho^t \in \mathcal{D}(\mathcal{H}^d)$, is a density matrix on the $d$-qubit, $2^d$-dimensional, Hilbert space $\mathcal{H}^d$.

Up to some unknown changepoint time $\nu\geq 1$, the state $\rho^t$ is any, unknown, state $\rho_0$ in a set $\mathcal{S}_0$; while, after that point, it corresponds to any, unknown, state $\rho_1$ in a different set $\mathcal{S}_1$ such that $\mathcal{S}_0\cap\mathcal{S}_1=\emptyset$. 
Thus, the state at time $t$ is given by
\begin{align}
\label{eq:sequence_of_states}
\rho^t =
\begin{cases}
	\rho_0\in\mathcal{S}_0, & t < \nu,\\
	\rho_1\in\mathcal{S}_1, & t \ge \nu.
\end{cases}
\end{align}

Specifically, the pre- and post-change state families $\mathcal{S}_0$ and $\mathcal{S}_1$ are characterized through constraints on the expectations of a collection of $n$ observables.  To elaborate, fix a set of $n$ observables $\{O_i\}_{i=1}^n$ whose expectation values determine the two regimes. Recall that an observable is a $2^d\times 2^d$ positive semi-definite matrix. We assume that any pre-change states $\rho_0\in\mathcal{S}_0$ satisfy the inequality
\begin{align}
	\langle O_i\rangle_{\rho_0} = \Tr(\rho_0 O_i) \le 0,
	\qquad \forall i \in \{1,\dots,n\},
\end{align}
while after the change, there exists at least one index $i$ for which the constraint is violated
\begin{align}
	\langle O_i\rangle_{\rho_1} = \Tr(\rho_1 O_i) > 0
	\quad \text{for some } i \in \{1,\dots,n\}.
\end{align}
For instance, the average energy of all particles encoding the qubits may be smaller than a threshold before the changepoint, while at least one particle exhibits a higher energy after the change. As another example, before the change, a set of entanglement witnesses \cite{terhal2000bell} would detect entanglement, while after the change, at least one witness would fail to detect entanglement, possibly indicating decoherence.

Equivalently, the pre- and post-change sets are given by 
\begin{align}
	\label{eq:pre_hypothesis}
	\mathcal{S}_0 &= \Big\{ \rho \in \mathcal{D}(\mathcal{H}^d) : 
	\max_{i \in\{1,\dots,n\}} \langle O_i\rangle_{\rho} \le 0 \Big\},
\end{align}
and
\begin{align}
	\label{eq:post_hypothesis}
	\mathcal{S}_1 &= \Big\{ \rho \in \mathcal{D}(\mathcal{H}^d) :
	\exists i\in\{1,\dots,n\} \text{ s.t. } \langle O_i\rangle_{\rho} > 0 \Big\}.
\end{align}

Note that the formulation in \eqref{eq:pre_hypothesis}--\eqref{eq:post_hypothesis} assumes, for simplicity, that the separating thresholds for each observable are zero. However, this entails no loss of generality, since any one-sided constraint of the form $\langle O_i\rangle_\rho \le a_i$ can be rewritten using the shifted observable $O_i’ = O_i - a_i I$. Similarly, a two-sided constraint $b_i \le \langle O_i\rangle_\rho \le a_i$ can be expressed through the shifted observables $O'_i = O_i - a_i I$ and $O''_i = b_i I - O_i$.

\emph{2) Universal Measurement Device:}   At each time $t$, based on a measurement strategy $\mu^t(\cdot)$, the measurement device chooses a POVM $\mathcal{M}^t = \{\Pi_x^t\}_{x \in \mathcal{X}}$ and measures the quantum state $\rho^t$. Accordingly, it observes a measurement outcome $X^t \in \mathcal{X}$ with probability distribution given by the Born rule
\begin{align}
	\label{eq:measurement_distribution}
	\Pr[X^t = x \mid \mathcal{M}^t, \rho^t] = \Tr(\Pi_x^t \rho^t).
\end{align}
The measurement strategy $\mu^t(\cdot)$ is allowed to be adaptive, as it may depend on the past measurement outcomes $X^{1:t-1} = (X^1,\dots,X^{t-1})$ and POVMs $\mathcal{M}^{1:t-1} = (\mathcal{M}^1,\dots,\mathcal{M}^{t-1})$. We denote the history of measurement outcomes and POVMs up to time $t$ as $H^{t} = (\mathcal{M}^{1:t},X^{1:t})$.

The measurement device is assumed to be agnostic to the observables $\{O_i\}^n_{i=1}$, supporting a universal measurement policy $\{\mu^t(\cdot)\}$ that may apply to an arbitrary set of observables. As illustrated in Figure \ref{fig:setting}, this setting allows for multiple changepoint detection defined by distinct observable sets to be performed by the same measurements. This universal property is important, since quantum measurements are destructive, and thus the quantum systems are no longer available in their original states $\{\rho^t\}_{t\geq 1}$ after measurement is applied.

\emph{3) Detection:}  Each measurement pair $(X^t,\mathcal{M}^t)$ is forwarded to a detection algorithm, which uses a decision rule $\phi^t(\cdot)$ to map the history $H^{t}$ to a detection outcome $D^t \in \{0,1\}$. The detection outcome $D^t = 0$ indicates that no changepoint has been detected so far, while $D^t = 1$ suggests that a changepoint has occurred, i.e.,
\begin{align}
    D^t=\phi^t(H^t)=\begin{cases}
        0,\quad  \text{ (no change detected)},\\
        1,\quad  \text{ (change detected)}.
    \end{cases}
\end{align}

The detection rule $\phi^t(\cdot)$ can be tailored to the observables $\{O_i\}^n_{i=1}$ that are known to the detector.

\emph{4) Performance Measures:}  In the following, for a fixed measurement policy $\{\mu^t(\cdot)\}_{t\geq 1}$ and decision rule $\{\phi^t(\cdot)\}_{t\geq 1}$, we use $\mathbb{P}_{\rho_0,\nu,\rho_1}[\cdot]$ and $\mathbb{E}_{\rho_0,\nu,\rho_1}[\cdot]$ to denote respectively the probability and expectation operators under quantum states with pre-change state $\rho_0$, changepoint time $\nu$, and post-change state $\rho_1$. Furthermore, if $\nu=\infty$, and no change occurs, we use the shorthand notation $\mathbb{P}_{\rho_0}[\cdot]$ and $\mathbb{E}_{\rho_0}[\cdot]$; while if $\nu=1$, and thus a change occurs immediately, we use $\mathbb{P}_{\rho_1}[\cdot]$ and $\mathbb{E}_{\rho_1}[\cdot]$.

Two main criteria that summarize the performance of a changepoint detection algorithm are its ability to rapidly detect changes after they occur, and its capacity to control the false detection rate. To formalize these objectives, let the first time the algorithm declares a change be denoted as
\begin{align}
	\label{eq:gen_stop_time}
	T = \min\{t \geq 1 : D^t = 1\}.
\end{align}
For a fixed post-change state $\rho_1$, the responsiveness of the detector with stopping time $T$ is then usually quantified by Lorden's metric, or worst-case expected detection delay \cite{lorden1971procedures}. The detection delay is defined as the difference $T-\nu$ for $T\geq \nu$, or more generally as $[T-\nu]_{+}$ with $[x]_{+}=x$ if $x\geq0$ and $[x]_{+}=0$ otherwise.   Considering the worst case across all possible pre-change states $\rho_0$,  changepoint locations $\nu$, and measurement outcomes $H^\nu$ before the changepoint, we obtain the worst-case expected detection delay
\begin{align}
	\label{eq:worst_worst_case_delay}
	\tau^*(\rho_1)=\hspace{-1em}\sup_{\rho_0\in\mathcal{S}_0,\nu\geq 0}\hspace{-1em}\esssup \mathbb{E}_{\rho_0,\nu,\rho_1}\left[[T-\nu]_{+}|H^\nu\right],
\end{align}
where the essential supremum accounts for the worst-case delay over the pre-change measurement history $H^\nu$, which follows the probability measure induced by the pre-change state $\rho_0$ and the measurement policy.

In contrast, the false-alarm rate requirement is conventionally specified via the worst-case average run length (ARL)
\begin{align}
	\label{eq:ARL}
	 \textrm{ARL}=\sup_{\rho_0\in\mathcal{S}_0}\mathbb{E}_{\rho_0}[T],
\end{align}
which represents the expected number of time steps before a false detection under the least favorable pre-change state. 

The goal is to design the universal measurement strategy $\{\mu^t(\cdot)\}_{t\geq1}$ and the observable-specific detection rule $\{\phi^t(\cdot)\}_{t\geq1}$ to minimize the detection delay $\tau^*(\rho_1)$ while satisfying an ARL requirement
\begin{align}
	  \label{eq:ARL_control}
      \textrm{ARL}\geq\frac{1}{\alpha},
\end{align}
where $\alpha\in (0,1)$ is a user-defined parameter.

\section{Universal Sequential Changepoint Detection of Quantum Observables via Classical Shadows}
\label{sec:escd}
In this section, we introduce eSCD, a changepoint detector for quantum observables that combines randomized measurements from the classical-shadow framework \cite{huang2020predicting} with Shiryaev–Roberts (SR) e-detectors \cite{shin2022detectors}. We first present the observable-agnostic measurement policy based on classical shadows, showing that it yields unbiased estimates of the target expected values of any set of target observables. We then describe the Shiryaev-Roberts (SR) detection rule, and establish non-asymptotic ARL control guarantees together with bounds on the worst-case expected detection delay. An alternative eSCD variant based on CUSUM e-detectors \cite{shin2022detectors} is presented in Appendix \ref{sec:cusum_detector}.

\subsection{Measurement Policy Based on Randomized Measurements}
\label{sec:classical_shadows}

To construct an observable-agnostic measurement strategy $\{\mu^t(\cdot)\}_{t\geq 1}$, we adopt the randomized measurement framework of classical shadows \cite{huang2020predicting}.
As shown in Figure \ref{fig:setting}, at each time $t$, we sample a $2^d\times 2^d$ unitary matrix $U^t \sim \mathcal{U}$ from a uniform distribution $\mathcal{U}$ over a suitable set of unitary matrices, and then apply standard measurements in the computational basis.  
We will specifically focus on the following typical choices for the set of unitary matrices \cite{huang2020predicting}:
\begin{itemize}
	\item \textbf{Local Clifford measurement}: With local Clifford measurements, the unitary matrix $U^t$ separates across qubits as
	$U^t = U_1^t \otimes \cdots \otimes U_d^t$, with each unitary matrix $U_k^t$ for each $k=1,\dots,d$ drawn uniformly from the
	single-qubit Clifford group. Local Clifford measurements can be hence equivalently implemented by choosing each matrix $U_k^t$ independently and uniformly from the set $\{I,H,HS^{\dag}\}$, where $H$ is the matrix of the Hadamard gate and $S$ is the matrix of the phase gate.
	\item \textbf{Joint Clifford measurement}:  With joint Clifford measurements, the unitary matrix
	$U^t$ is drawn uniformly from the full $d$-qubit Clifford group. This can be achieved using, for example, the symplectic matrix construction \cite{koenig2014efficiently}. Note that, unlike the local Clifford ensemble, implementing the unitary $U^t$ requires multi-qubit operations and deeper circuits. 
\end{itemize}

Performing a computational-basis measurement of the rotated state $U^t \rho^t (U^t)^\dagger$ corresponds to implementing the POVM
\begin{align}
	\label{eq:shadow-POVM-general}
	\mathcal{M}^t 
	= \Big\{\Pi_x^t=(U^t)^\dagger \ket{x}\bra{x} U^t \Big\}_{x\in\{0,1\}^d}.
\end{align}
By the Born rule \eqref{eq:measurement_distribution}, this measurement returns as the classical outcome the $d$-bit string
\begin{align}
	\label{eq:shadow_outcome_meas}
	X^t \in \{0,1\}^d\sim\Pr[X^t=x|\mathcal{M}^t,\rho^t]=\Tr(\Pi_x^t\rho^t),
\end{align}
which using \eqref{eq:shadow-POVM-general} can also be expressed as
\begin{align}
	\Tr(\Pi_x^t\rho^t)=\bra{x} U^t\rho^t(U^t)^\dagger \ket{x}.
\end{align}
For a state $\rho$ and randomized measurement $\mathcal{M}$ with distribution $\mathcal{U}$, define the measurement channel
\begin{align}
	\label{eq:measurement_channel}
	\mathcal{R}(\rho)
	&= \mathbb{E}_{U,X}\left[U^\dagger \ket{X}\bra{X} U\right] \nonumber\\
	&= \mathbb{E}_{U}\left[ \sum_{x\in\{0,1\}^d} \bra{x}U\rho U^\dagger\ket{x}\; U^\dagger \ket{x}\bra{x} U \right],
\end{align}
where $\mathbb{E}_{U,X}[\cdot]$ represents the average over the variables $U\sim \mathcal{U}$ and $X\sim \Pr[X=x|\mathcal{M},\rho]$. The map $\mathcal{R}(\rho)$ is linear and invertible for both the local-Clifford and the joint-Clifford ensembles \cite{huang2020predicting}.

For a given $2^d\times 2^d$ unitary $U$ and measurement outcome $X=(X_1,\dots,X_d)$ with $X_k\in\{0,1\}$ for $k=1,\dots,d$ define the snapshot
\begin{align}
    \label{eq:snapshot}
	\sigma=U^\dagger\ket{X}\bra{X}U,
\end{align}
which, in the case of a local Clifford unitary matrix $U=U_1\otimes\dots\otimes U_d$ factorizes as
\begin{align}
	\label{eq:snapshot_single_Cliff}
	\sigma =\bigotimes_{k=1}^d U_k^\dagger\ket{X_k}\bra{X_k}U_k .
\end{align}
Furthermore, with local Clifford measurements, the inverse of the measurement channel
\eqref{eq:measurement_channel} evaluated at the snapshot $\sigma$ in \eqref{eq:snapshot_single_Cliff} has the form \cite{huang2020predicting}
\begin{align}
	\label{eq:Pauli-inverse-map}
	\hat{\rho} = \mathcal{R}^{-1}(\sigma)
	= \bigotimes_{k=1}^d \left( 3 U_k^\dagger\ket{X_k}\bra{X_k}U_k - I \right).
\end{align}
In contrast, with joint Clifford measurements, the inverse channel evaluated at the snapshot $\sigma$ in \eqref{eq:snapshot} is given by \cite{huang2020predicting}
\begin{align}
	\label{eq:joint-inverse-map}
	\hat{\rho} = \mathcal{R}^{-1}(\sigma)
	= (2^d + 1)\sigma - I.
\end{align}
\subsection{Detection: State Estimates based on Classical Shadows}
At each time $t$, the detector receives the transformation $U^t$ and the measurement output $X^t$ in \eqref{eq:shadow_outcome_meas}. Following the classical shadow protocol, it evaluates the snapshot \eqref{eq:snapshot} as
\begin{align}
	\sigma^t = (U^t)^\dagger \ket{X^t}\bra{X^t} U^t.
\end{align}
Then it applies the inverse of the measurement channel \eqref{eq:measurement_channel} to produce the so-called classical shadow estimate
\begin{align}
	\hat\rho^t = \mathcal{R}^{-1}(\sigma^t).
\end{align}
The classical shadow $\hat\rho^t$ is an unbiased estimator of the true state in the sense that, when averaged over the randomized measurement and measurement outcome, the true state $\rho^t$ is recovered as
\begin{align}
    \label{eq:unbiased_shadow}
	\mathbb{E}_{U^t,X^t}[\hat{\rho}]
	&= \mathcal{R}^{-1}\left(\mathbb{E}_{U^t,X^t}[\sigma^t]\right)
	= \mathcal{R}^{-1}\left(\mathcal{R}(\rho^t)\right)
	= \rho^t.
\end{align}
where we have used the definition of measurement channel \eqref{eq:measurement_channel}. The unbiasedness property \eqref{eq:unbiased_shadow} of the classical shadow $\hat{\rho}^t$ implies that
for any observable $O$ the quantity
\begin{align}
		\label{eq:o_hat_i}
	\hat{o}^t = \Tr\left(O\hat{\rho}^{t}\right),
\end{align}
represent an unbiased estimate of the expected value $\langle O_i\rangle_{\rho^t}=\Tr(O\hat{\rho^t})$. This result is formalized in the following proposition. Recall that the operator norm $\lVert O\rVert_{\infty}$ of a positive semidefinite matrix $O$ is defined as the largest eigenvalue of $O$, i.e. $\lVert O\rVert_{\infty}=\lambda_{\rm{max}}(O)$.
\begin{proposition}
	\label{prop:bound_o_hat_general}
	Let $O$ be an observable with finite operator norm, i.e., $\|O\|_\infty < \infty$. Then, the estimator $\hat{o}^t$ in \eqref{eq:o_hat_i} is unbiased, i.e., $\mathbb{E}_{U^t,X^t}[\hat{o}^t] = \langle O \rangle_{\rho^t}$,	and there exist finite constants $l \le u$ such that $\hat{o}^t\in[l, u]$ almost surely.
\end{proposition}
\begin{proof}
    See Appendix \ref{app:bound_o_hat_general}.
\end{proof}  

\subsection{Detection: Betting-based Decision Rule}
\label{sec:e-SB-SR-detector}
Given the observables $\{O_i\}_{i=1}^n$ defining the pre- and post-change regimes, at each round $t\ge 1$, the detector produces the set of estimates $\{\hat{o}_i^t\}^n_{i=1}=\{\Tr(O_i\hat\rho_t)\}^n_{i=1}$ in \eqref{eq:o_hat_i}. To turn the sequence of estimates $\{\{\hat{o}_i^t\}_{t\geq1}\}^n_{i=1}$ into a changepoint detection algorithm with ARL control, we rely on the framework of e-detectors introduced in \cite{shin2022detectors}.  

Consider a sequential game in which, at each time step $t$, a bettor invests part of its current capital to bet on the value of the estimate $\hat{o}^t_i$. Let $\lambda^t_i$ denote the variable indicating the size of the bet, with its sign determining whether the bet is placed in favor of or against a positive mean of $\hat{o}^t_i$. The baseline increment process $\{L^t_i\}_{t\geq1}$, given by
\begin{align}
	\label{eq:baseline_increment}
	L_i^t=1+\lambda^t_i\hat{o}^t_i,
\end{align}
represents the one-step capital multiplier of the bettor. The betting parameter $\lambda_i^t$ is an $H^{t-1}$-predictable betting parameter, meaning that it is chosen based only on the history $H^{t-1} = (X^{1:t-1}, \mathcal{M}^{1:t-1})$, and it satisfies $\lambda_i^t\in \Lambda_i = (-u_i^{-1},\, -l_i^{-1})$ where $l_i$ and $u_i$ are the minimum and maximum value of $\hat{o}^t_i$, respectively.

The baseline increments \eqref{eq:baseline_increment} are accumulated over time index $t$ according to the multiplicative rule \cite{shiryaev1961problem} 
\begin{align}
	\label{eq:per_observable_e-detector}
	M^t_{ \rm{SR},i}&=L_i^t\cdot(M^{t-1}_{ \rm{SR},i}+1)\nonumber\\
	&=\sum^t_{j=1}\prod^t_{k=j}L^j_i\nonumber\\
    &=\sum^t_{j=1}E_i^{j:t},
\end{align}
with $M^0_{\rm{SR},i}=0$, and where $E_i^{j:t}$ denotes the cumulative product of the baseline increments process $\{L^t_i\}_{t\geq1}$ from time $j$ to time $t$, with $j\leq t$, i.e.,
\begin{align}
	\label{eq:baseline_partial_prods}
	E_i^{j:t} = \prod_{k=j}^t L_i^k.
\end{align}
The statistic \eqref{eq:per_observable_e-detector} is an instance of the SR e-detector \cite{shiryaev1961problem} introduced in \cite{shin2022detectors}.

Note that the statistic \eqref{eq:baseline_partial_prods} admits a natural betting interpretation as it corresponds to the capital of a bettor who enters the game at time $j$ and follows the betting strategy $\{\lambda_t\}_{t\geq 1}$. This capital is expected to grow large if the mean of the observable $O_i$ after time $j$ is positive. Consequently, the SR statistic \eqref{eq:per_observable_e-detector} aggregates evidence across all possible changepoints by summing the capital processes of bettors entering at different times, and for the same reasons, it is expected to grow large after the unknown changepoint $\nu$.

The per-observable SR e-detector $\{M^t_{\rm{SR},i}\}^n_{i=1}$ are aggregated to obtain an e-detector for the pre-change set $\mathcal{S}_0$ defined based on the set of all observables $\{O_i\}^n_{i=1}$. In particular, for a weight vector $\mathbf{w}=[w_1,\dots,w_n]$ such that $w_i>0$ and $\lVert \mathbf{w}\rVert_1=1$, we define the SR e-detector as the weighted sum
\begin{align}
	\label{eq:e-detector}
	M^t_{\rm{SR}}&=\sum^n_{i=1}w_i	M^t_{\rm{SR},i}=\sum^n_{i=1}w_i	\sum^t_{j=1}E^{j:t}_i.
\end{align}

At each time step $t$, the decision rule 
\begin{align}
	\label{eq:decision_rule}
	D^t=\phi^t(M^t_{\rm{SR}})=\begin{cases}
		1, \quad \text{ if } 	M^t_{\rm{SR}}>1/\alpha,\\
		0, \quad \text{ otherwise.}
	\end{cases}
\end{align}
detects a changepoint if the e-detector $M^t_{\rm{SR}}$ takes a value larger than a threshold, which is set to $1/\alpha$. 

The decision rule \eqref{eq:decision_rule} applied by eSCD induces the stopping time \eqref{eq:gen_stop_time} as
\begin{align}
	\label{eq:stopping_times_aggr_1}
	T_{\rm{eSCD}}&=\inf\{t\geq1: M^t_{\rm{SR}}\geq 1/\alpha\}.
\end{align}

\subsection{Theoretical Guarantees}

In the following, we establish theoretical guarantees for eSCD. First, we prove that it satisfies the ARL control requirement in \eqref{eq:ARL_control} for any predictable choice of betting parameters and mixture weights. Then, we provide efficiency guarantees by deriving bounds on the worst-case expected detection delay \eqref{eq:worst_worst_case_delay} and establishing that betting strategies with sublinear strongly adaptive regret are asymptotically optimal within eSCD's class of changepoint detectors.
\subsubsection{ARL Control}
For any fixed weight vector $\mathbf{w}$ and any predictable sequences of betting parameters $\{\{\lambda_i^t\}_{t\geq1}\}^n_{i=1}$, eSCD satisfies the ARL control requirement \eqref{eq:ARL_control}.
\begin{theorem}[ARL Control]
	\label{th:arl_control}
	For any collection of observables $\{O_i\}_{i=1}^n$, weight vector $\mathbf{w}$, predictable betting parameters $\{\{\lambda_i^t\}_{t\geq1}\}^n_{i=1}$, and threshold 
	$\alpha\in(0,1)$, eSCD satisfies the ARL requirement \eqref{eq:ARL_control}, i.e,
	\begin{align}
		\sup_{\rho_0\in\mathcal{S}_0}\mathbb{E}_{\rho_0}[T_{\rm{eSCD}}]\geq \frac{1}{\alpha}.
	\end{align}
\end{theorem}
The proof of Theorem \ref{th:arl_control} is given in Appendix \ref{proof:arl_control}, and it follows from the unbiasedness of the estimates provided by the randomized measurement protocol established in Proposition \ref{prop:bound_o_hat_general} and from the ARL guarantees of e-detectors \cite{shin2022detectors}. Note that, by virtue of the properties of e-detectors, the ARL control guarantee in Theorem \ref{th:arl_control} can be generalized beyond a fixed pre-change quantum state $\rho_0$. Specifically, in Appendix \ref{proof:arl_control}, we prove the more general result showing that eSCD guarantees ARL control for any sequence of quantum states $\{\rho^t\}_{t\geq 1}$ such that $\rho^t\in \mathcal{S}_0$. 

\subsubsection{Efficiency}
Theorem \ref{th:arl_control} guarantees ARL control for any choice of mixture weights $\mathbf{w}$ and any sequence of betting parameters $\{\{\lambda_i^t\}_{t\geq1}\}_{i=1}^n$; however, the efficiency of the eSCD, as captured by the worst-case average delay $\tau^*(\rho_1)$ in \eqref{eq:worst_worst_case_delay}, depends on these hyperparameters. To quantify this dependence, it is common to study the growth of the test statistic \eqref{eq:per_observable_e-detector} after the changepoint since this dictates how quickly the decision rule \eqref{eq:stopping_times_aggr_1} identifies the change. For reference, following \cite{shin2022detectors}, we first consider the growth achievable by an oracle eSCD detector that knows the post-change state $\rho_1$. For a post-change state $\rho_1\in\mathcal{S}_1$ and observable $O_i$, define the maximal expected log-increment of the baseline process $\{L_i^t\}_{t\geq1}$ as
\begin{align}
	D^*_i(\rho_1)
	=
	\max_{\lambda_i\in\Lambda_i}
	\mathbb{E}_{\rho_1}
	\bigl[\log\bigl(1+\lambda_i \hat{o}_i^1\bigr)\bigr],
\end{align}
and denote by $\lambda_i^*$ a maximizing betting parameter. Among all observables, the one inducing the largest expected log-growth rate is
\begin{align}
	\label{eq:maximum_log_growth}
	D^*(\rho_1)=\max_{i\in[n]}D^*_i(\rho_1),
\end{align}
and we write $i^*$ for an index attaining this maximum. 

If the post-change state $\rho_1$ were known in advance, the hyperparameters that maximize the expected log-growth of the test statistic in \eqref{eq:e-detector} would set $\mathbf{w}$ equal to the unit vector assigning all mass to observable $i^*$, and would use the constant betting strategy $\lambda_i^t=\lambda_i^*$ for all $t$. However, since $\rho_1$ is unknown, a central challenge is to design a sequence of betting parameters and mixture weights that performs nearly as well as this oracle choice for every $\rho_1\in\mathcal{S}_1$. 

In the context of sequential hypothesis testing, numerous works have connected the power of sequential tests to regret guarantees of online learning algorithms, showing that procedures with sublinear static regret give rise to asymptotically optimal testing algorithms in terms of their expected stopping time \cite{10220229,10315047,RAMDAS202283,casgrain2024sequential,waudby2025universal,agrawal2025stopping}. Unlike sequential hypothesis testing, in changepoint detection the time for which one wishes to ensure the growth of the test statistic is unknown, this corresponds to the changepoint $\nu$. Accordingly, in this section, we show that asymptotic optimality can be obtained by betting strategies targeting a dynamic notion of regret, namely, strongly adaptive regret \cite{hazan2007adaptive,hazan2009efficient}.

More precisely, define the maximum log-growth of the cumulative product of the baseline increments process $\{L^t_{i^*}\}_{t\geq1}$ in any interval $[j,t] =\{j,j+1,\dots,t\}$ 
\begin{align}
	\max_{\lambda\in\Lambda_{i^*}}\log(E^{j:t}_{i^*})=\max_{\lambda\in\Lambda_{i^*}}\sum^t_{k=j}	\log\bigl(1+\lambda_i \hat{o}_i^t\bigr).
\end{align}
Note that the index $j$ plays the role of the unknown changepoint time $\nu$. We then define the regret over this time interval, and the property of sublinear strongly adaptive regret of a betting strategy as follows. The regret of a hyperparameter sequence $(\mathbf{w},\{\{\lambda_i^t\}_{t\geq1}\}_{i=1}^n)$ in the interval $[j,t] $ is
\begin{align}
	\label{eq:time_int_regret}
	\mathcal{R}^{j:t}
	=
	\max_{\lambda\in\Lambda_{i^*}}\log E_{i^*}^{j:t}(\lambda)
	-
	\log\left( \sum_{i=1}^n w_i E_i^{j:t} \right),
\end{align}
where $E_{i}^{j:t}(\lambda)$ denotes the cumulative product of the baseline increment process from time $j$ to $t$, evaluated under a constant betting scheme with $\lambda^t_i=\lambda$.
Furthermore, for any horizon $T$, the strongly adaptive regret is defined as the maximum value of \eqref{eq:time_int_regret} over all intervals of length $\Delta_T$,
\begin{align}
	\label{eq:sa_regret}
	\text{SAR}(T,\Delta_T)
	=\max_{[j,j+\Delta_T]\in[1,T]}	\mathcal{R}^{j:j+\Delta_T}.
\end{align}
Different optimization algorithms admit sublinear strongly adaptive regret for convex objectives \cite{daniely2015strongly,wang2018minimizing,jun2017improved,cutkosky2020parameter}. In Section \ref{sec:betting_strategy} we derive a betting strategy for eSCD based on the coin-betting framework \cite{jun2017improved} that enjoys a sublinear bound on the strongly adaptive regret \eqref{eq:sa_regret}.

We now show that eSCD instantiated with a sequence of hyperparameters  $(\mathbf{w},\{\{\lambda_i^t\}_{t\geq1}\}_{i=1}^n)$ that achieves sublinear strongly adaptive regret, enjoys the following asymptotic worst-case expected detection delay guarantee.

\begin{theorem}[Efficiency Guarantee]
\label{th:eff_theorem}
 For a weight vector  $\mathbf{w}$ and predictable sequences of betting parameters $\{\{\lambda_i^t\}_{t\geq1}\}^n_{i=1}$ that enjoy sublinear strongly adaptive regret \eqref{eq:sa_regret}, eSCD satisfies
\begin{align}
	\label{eq:optimal_detection_delay}
	\lim_{\alpha\to0^+}\frac{\tau^*(\rho_1)}{\log(1/\alpha)}=\frac{1}{D^*(\rho_1)}.
\end{align}
\end{theorem}

\subsection{Coin Betting for Changing Environments Betting Scheme}
\begin{figure}
    \centering
    \includegraphics[width=0.95\linewidth]{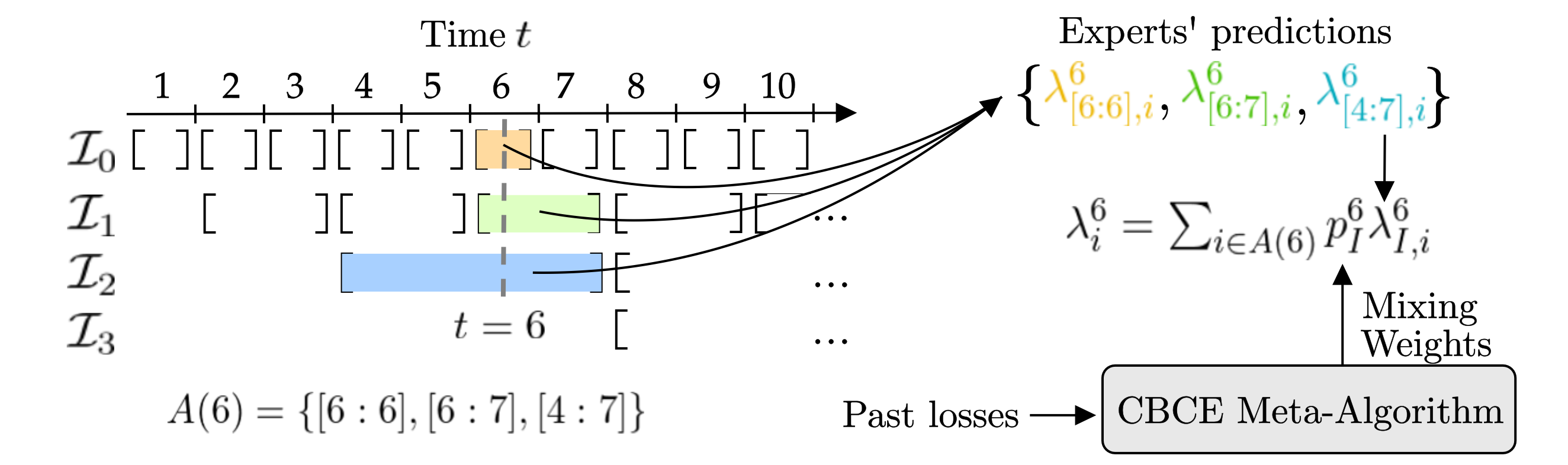}
    \caption{Illustration of the betting scheme for observable $O_i$ based on the Coin Betting for Changing Environments (CBCE) algorithm \cite{jun2017improved}. The time horizon is partitioned into sub-intervals of doubling length, and a universal portfolio algorithm is instantiated within each interval. At each time $t$, the active experts are those whose intervals contain $t$. The CBCE prediction is formed as a convex combination of their outputs, with weights determined by their past losses through the CBCE meta-algorithm \cite[Algorithm 2]{jun2017improved}. }
    \label{fig:cbce}
\end{figure}
\label{sec:betting_strategy}
As discussed, eSCD instantiated with a sequence of hyperparameters that achieves a sublinear strongly adaptive regret, enjoys the asymptotic worst-case expected detection delay in Theorem \ref{th:eff_theorem}. In this subsection, we describe a betting scheme that achieves this guarantee by leveraging the Coin Betting for Changing Environments (CBCE) meta-algorithm \cite{jun2017improved}, instantiated with universal portfolio experts \cite{cover1991universal}. This will allow us to refine the asymptotic result \eqref{eq:optimal_detection_delay} by specializing it to a state-of-the-art betting algorithm with strongly adaptive regret guarantees.

As illustrated in Figure \ref{fig:cbce}, CBCE is a meta-learning algorithm that operates by dividing the optimization horizon into overlapping sub-intervals and by instantiating independent ``experts" across each interval. Each expert corresponds to an online convex optimization instance with sublinear static regret. At each prediction step $t$, CBCE monitors the performance of all active experts and aggregates their predictions to obtain a meta-prediction that enjoys sublinear strongly adaptive regret \cite{jun2017improved}.

In detail, following \cite{daniely2015strongly}, the time axis $t \in \mathbb{N}$ is divided into collections of intervals $\mathcal{I}_k$, with each collection $\mathcal{I}_k$ containing all intervals of size $2^k$ time steps. Specifically, the $k$-th partition encompasses the intervals 
\begin{align}
    \mathcal{I}_k = \big\{[i2^k, (i+1)2^k - 1] : i \in \mathbb{N}\big\}.
\end{align}
At each time $t$, let $A(t)$ denote the set of intervals in the union $\mathcal{I}=\cup^\infty_{k=1}\mathcal{I}_k$ that contain $t$, i.e.,
\begin{align}
	A(t) = \{ I \in \mathcal{I} : t \in I \}.
\end{align}
Note that the set $A(t)$ is finite and contains $|A(t)| = \lfloor \log_2(t) \rfloor + 1$ intervals.

For each observable $i$ and each interval $I = [t_1, t_2] \in \mathcal{I}$, CBCE instantiates a universal portfolio expert \cite{cover1991universal}. For every time $t_1 < t \leq t_2$ in the interval $I$, the universal portfolio expert outputs the betting parameter
\begin{align}
    \label{eq:expert_bet}
	\lambda^t_{I,i}
	= \frac{\displaystyle \int_{\Lambda_i} \lambda E^{t_1:t}_i(\lambda)\mathrm{d}F_i(\lambda)}%
	{\displaystyle \int_{\Lambda_i}E^{t_1:t}_i(\lambda)\mathrm{d}F_i(\lambda)},
\end{align}
where $F_i(\lambda)$ represents the $\mathrm{Beta}(1/2, 1/2)$ distribution rescaled to the interval $\Lambda_i$.

The predictions of the experts in $A(t)$ are combined using weights $\{p_I^t\}_{I \in A(t)}$, with $p_I^t\geq 0$ and $\sum_{I \in A(t)} p_I^t=1$, as
\begin{align}
	\label{eq:cbce_betting_params}
	\lambda^t_{i}
	= \sum_{I \in A(t)} p_I^t\lambda^t_{I,i}.
\end{align}
The weights $\{p_I^t\}_{I \in A(t)}$ are computed using the CBCE meta-algorithm \cite[Algorithm 2]{jun2017improved}.

The combination of hyperparameters $(\mathbf{w},\{\{\lambda_{i}^t\}_{t\geq1}\}_{i=1}^n)$ with bets generated under the CBCE betting strategy can be shown to enjoy a sublinear strongly adaptive regret \eqref{eq:sa_regret}. Furthermore, leveraging this guarantee, the following corollary provides a non-asymptotic bound on eSCD's worst-case detection delay under a CBCE betting strategy that recovers the asymptotic guarantee.

\begin{corollary}[Efficiency Guarantee under CBCE betting]
	\label{prop:sublinear_regret}
	For a non-degenerate weight vector $\mathbf{w}$ such that $w_i > w_{\rm{min}}>0$ for all $i$, and sequence of bets $\{\{\lambda^t_{i}\}_{t\geq1}\}_{i=1}^n$ generated according to \eqref{eq:cbce_betting_params}, there exists a constant $C<\infty$ such that for any $\delta\in(0,1)$ and integer $s>2$, eSCD's worst-case detection delay satisfies 
	\begin{align}
        \label{eq:finite_sample_det_delay}
    	\tau^*(\rho_1)\leq 3&+ \frac{(1+\delta)\log(1/\alpha)}{D^*(\rho_1)}+\frac{2^s \sigma_s(\rho_1)}{(s/2-1)}\left(\frac{1+\delta}{D^*(\rho_1)}+\frac{1}{\log(1/\alpha)}\right)\nonumber\\
    	&+\frac{2(1+\delta)\alpha^{\delta/2}}{\delta D^*(\rho_1)}+\sum^\infty_{k=m+\nu}\mathds{1}\left\{C\sqrt{\frac{(7\ln(k)+5)}{k-\nu}}+\frac{\log(1/w_{\rm{min}})}{k-\nu}\geq \frac{\epsilon}{2}\right\},
    \end{align}
	where $m=\left\lceil(1+\delta)\log(1/\alpha)/D^*(\rho_1)\right\rceil$, $\epsilon=\delta D^*(\rho_1)/(1+\delta)$ and $\sigma_s(\rho_1)$ is the $s$-th central moment of the log-increment under $\rho_1$,
 	\begin{align}
        \label{eq:centr_log_moment}
		\sigma_s(\rho_1)
		=
		\mathbb{E}_{\rho_1}\Bigl[\bigl|\log\bigl(E_{i^*}^{1}(\lambda^*_{i^*})\bigr)-D^*(\rho_1)\bigr|^s\Bigr].
    \end{align}
\end{corollary}

In Appendix \ref{proof:eff_theorem} we show that the central moment of the log-increment is bounded. Note that for $\alpha \to 0^+$ and any $\delta\in(0,1)$ the upper bound \eqref{eq:finite_sample_det_delay} scales as $(1+\delta)\log(1/\alpha)/D^*(\rho_1)$ which can be made arbitrarily close to $\log(1/\alpha)/D^*(\rho_1)$ taking $\delta \to 0$. This yields the asymptotic guarantee \eqref{eq:optimal_detection_delay}. In this regime, $m \to \infty$ and the dominant term becomes $\frac{(1+\delta)\log(1/\alpha)}{D^*(\rho_1)}$, while all other terms remain constant or vanish. In particular, for sufficiently large $m$, all terms in the series become zero.

\section{Experiments}
\label{sec:experiments}
In this section, we present a series of experiments to evaluate eSCD’s ARL control and detection performance, and to compare eSCD with changepoint detection schemes that use matched measurement policies tailored to the specific changepoint under test. Unless stated otherwise, eSCD is implemented using local Clifford measurements.
\subsection{Single Observables}
\label{sec:exp_single_obs}
We first evaluate eSCD by considering changepoints defined by a single non-trivial, i.e., traceless, Pauli string $O=X^{\otimes d}$.
Specifically, we simulate a sequence of $d$-qubit quantum states $\{\rho^t(\theta)\}_{t\geq 1}$
parameterized by $\theta^t \in [-1,1]$ and defined as
\begin{align}
	\label{eq:exp_qubit}
	\rho(\theta^t) = \frac{I + \theta^t X^{\otimes d}}{2^d},
\end{align}
such that the pre-change set $\mathcal{S}_0$ in \eqref{eq:pre_hypothesis} contains states with $\theta^t \leq 0$, while the post-change set $\mathcal{S}_1$ in \eqref{eq:post_hypothesis} is characterized by $\theta^t > 0$. Note that setting $\theta^t = 0$ yields the maximally mixed state $\rho(0) = I/2^d$, while the choice $\theta^t = \pm 1$ corresponds to states that are maximally mixed within the subspaces associated with the positive ($\theta^t = 1$) or negative ($\theta^t = -1$) eigenvalues of the observable $O$. 

As a benchmark, we consider a measurement strategy tailored to the observable $O$. More specifically, let $\{\ket{v_x}\}_{x=1}^{2^d}$ and $\{o_x\}_{x=1}^{2^d}$ denote the eigenvectors and eigenvalues of the observable $O$, respectively. At each time $t$, the projective measurement $\mathcal{M} = \{\ket{v_x}\bra{v_x}\}_{x=1}^{2^d}$ is applied, and the observed outcome $x^t$ are used to instantiate the e-detector \eqref{eq:per_observable_e-detector} with the baseline increment \eqref{eq:baseline_increment} evaluated at $\hat{o}_i^t=o_{x^t}$. We refer to this scheme as the matched measurement-based changepoint e-detector (eMCD).

In the left panel of Figure \ref{fig:cbce_delay}, we evaluate the performance in terms of ARL by setting the number of qubits $d=2$, the changepoint to $\nu = \infty$ and considering a pre-change state $\rho_0(\theta^t) \in \mathcal{S}_0$ with parameter $\theta^t = \theta_0$. We report the smoothed empirical distribution and mean value of the run length \eqref{eq:stopping_times_aggr_1}, averaged over 100 runs, for a pre-change parameter $\theta_0$ that varies from $0.5$ to $0$. We set the ARL requirement to $1/\alpha = 1000$, and cap the run length to $5000$. For all values of $\theta_0$, thanks to the properties of e-detectors \cite{shin2022detectors}, the ARL requirement \eqref{eq:ARL_control} is satisfied. Moreover, as the parameter $\theta_0$ increases, and thus the quantum state approaches the set of post-change quantum states (with $\theta^t>0$), the ARL decreases. We also observe that, for the same ARL requirement, eMCD yields a larger run length. This quantifies the cost of universality of eSCD with respect to the observable-matched strategy applied by eMCD.

In the right panel of Figure \ref{fig:cbce_delay}, we evaluate the efficiency of the two detection algorithms by considering an unknown but fixed changepoint at $\nu = 200$. Before the changepoint $\nu$, we set the pre-change parameter to $\theta^t = \theta_0 = -0.5$ for $t < \nu$, and after the change, we set the parameter to $\theta^t = \theta_1 > 0$ for $t \ge \nu$. For a fixed ARL requirement of $1/\alpha = 1000$, we report the smoothed empirical distribution of the detection delay averaged over $100$ runs as a function of the post-change parameter $\theta_1$, which ranges from $0.1$ to $1$. As the value of the post-change parameter $\theta_1$ grows, the violation of the observable expectation becomes more severe, and the detection delay decreases. In general, eMCD, being tailored to the observable $O$, delivers smaller detection delays than eSCD, which instead applies universally to any observable. As illustrated in the next experiment, the price for universality paid by eSCD highlighted by this example is quickly amortized as the number of observables $n$ grows.

\begin{figure}[h!]
	\centering
	\begin{subfigure}{0.45\textwidth}
		\centering
		\includegraphics[width=\textwidth]{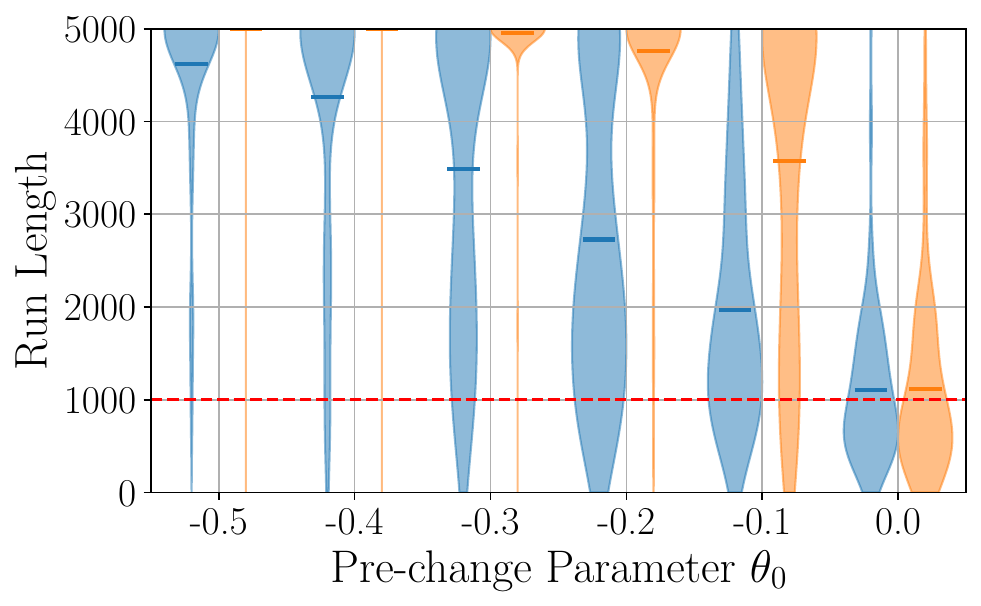}
		\label{fig:cbce_delay_a}
	\end{subfigure}\hfill
	\begin{subfigure}{0.45\textwidth}
		\centering
		\includegraphics[width=\textwidth]{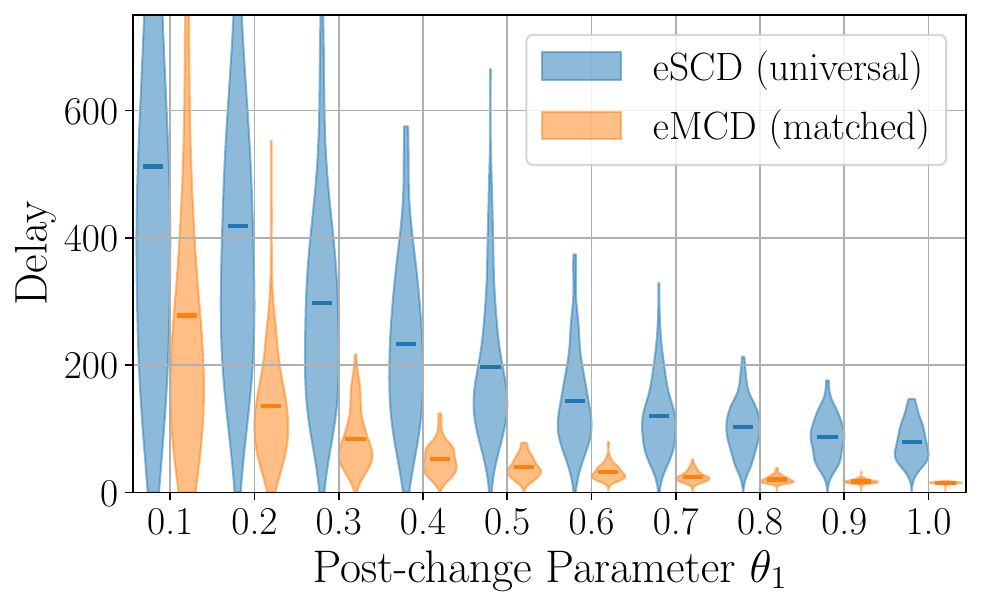}
		\label{fig:cbce_delay_b}
	\end{subfigure}
	\caption{Violin plot showing the distribution of the ARL (left) and the detection delay (right) for eSCD and eMCD as a function of the pre-change parameter $\theta_0$ (left) and as a function of the post-change parameter $\theta_1$ for $\nu=200$ and $\theta_0=-0.5$ (right). The ARL requirement is set to  $1/\alpha = 1000$ (red, dashed line, left), means are shown as straight segments within each violin plot, and quantities are averaged over 100 runs.}
	\label{fig:cbce_delay}
\end{figure}

\subsection{Multiple Observables}

\begin{figure}
	\centering
	\includegraphics[width=0.95\textwidth]{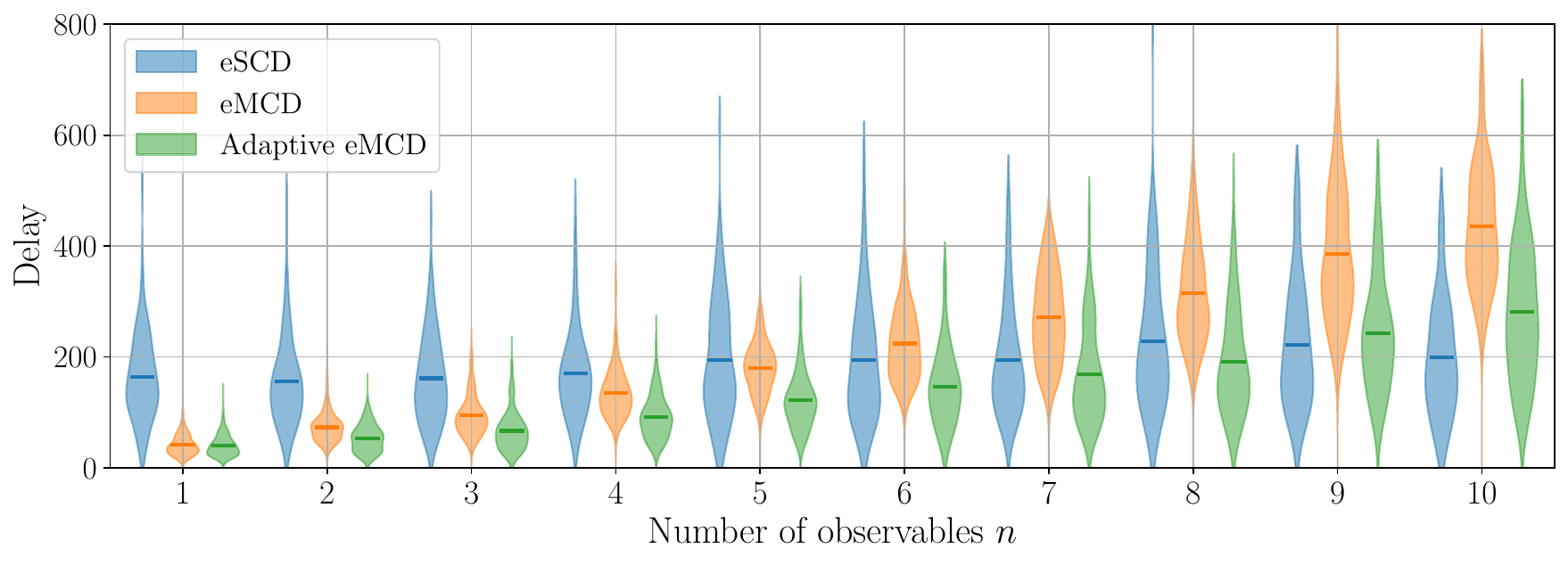}
	\caption{Violin plots for the distribution of the detection delay for eSCD, eMCD, and adaptive eMCD as a function of the number of observable $n$. Quantities are averaged over 100 runs, and all schemes target the ARL constraint $1/\alpha=1000$.}
	\label{fig:multiple_observables}
\end{figure}
 
We now consider changepoints defined by a collection of $n\geq1$ observables $\{O_i\}_{i=1}^n$. For any integer $n \ge 1$, the observables correspond to $n$ rotated variants of the base observable $O = \bigotimes_{j=1}^d X$. Specifically, each observable $O_i$ is obtained by rotating $O$ about the $z$-axis by an angle $\gamma_i = \pi i/(2n)$, i.e.,
\begin{align}
	\label{eq:rotated_observable}
	O_i= (R_z^\dagger(\gamma_i) X R_z(\gamma_i))^{\otimes d},
\end{align}
with Pauli-Z rotation
\begin{align}
	R_z(\gamma)
	= e^{- i \gamma Z / 2}
	= \cos\left(\frac{\gamma}{2}\right) I
	- i \sin\left(\frac{\gamma}{2}\right) Z .
\end{align}

We simulate a sequence of quantum states $\{\rho^t\}_{t\geq 1}$ that evolves as
\begin{align}
	\rho^t=\begin{cases}
		\rho_0=\frac{I -0.5  X^{\otimes d}}{2^d}, \quad \text{for } t<\nu,\\
		\rho_1=\frac{I +0.5  X^{\otimes d}}{2^d}, \quad \text{for } t\geq\nu.
	\end{cases}
\end{align}
The instant $\nu$ is a valid changepoint for the observable \eqref{eq:rotated_observable} since, for all $i\in\{1,\dots,n\}$, we have the inequalities
\begin{align}
	\Tr(\rho_0O_i)=-0.5\cos\left(\frac{\pi i}{2n}\right)<0, \text{ and }\Tr(\rho_1O_i)=0.5\cos\left(\frac{\pi i}{2n}\right)\geq0.
\end{align}

We use this set-up to investigate the potential drawbacks and benefits of eSCD as compared to a generalized eMCD that tailors its measurements to the observables $\{O_i\}_{i=1}^n$.  In particular, at each time $t$, eMCD selects an index $i^t \in \{1,\dots,n\}$, and performs the projective measurement corresponding to the observable $O_i^t$. The measurement outcome $x^t$ is then used to update the corresponding per-observable e-detector $E^t_{i^t}$ in \eqref{eq:per_observable_e-detector} while all other per-observable detectors remain unchanged, so that $E^t_j = E^{t-1}_j$ for all $j \neq i^t$.

We consider two instantiations of eMCD that apply different policies for selecting the index $i^t$ at each time $t$. 

$\bullet$ The first is a uniform policy, where $i^t$ is chosen in a round-robin fashion cycling through $\{1,\dots,n\}$.  This yields an eMCD scheme that uses measurements tailored to the specific observables $\{O_i\}^t_{i=1}$ but is non-adaptive, in the sense that the observable index $i_t$ does not depend on past measurement outcomes.

$\bullet$  We also consider an adaptive measurement policy based on an upper confidence bound (UCB) applied to the baseline increments. Let $N_t(i)$ denote the number of times index $i$ has been selected up to time $t$, and let $\{L^n_i\}_{n=1}^{N_t(i)}$ be the corresponding sequence of baseline increments for that observable that has been obtained up to time $t$. 
At time $t+1$, the index is chosen as
\begin{align}
	i_{t+1}
	= \argmax_j 
	\left(
	\frac{1}{N_t(j)} \sum_{n=1}^{N_t(j)} L_j^n
	\;+\;
	\sqrt{\frac{2\log(1/\delta)}{N_t(j)}}
	\right),
\end{align}
where $\delta \in (0,1)$ is the confidence parameter. We refer to the resulting detector as the adaptive eMCD. Its measurement policy is not only tailored to the observable family defining the changepoint but also adapts to the past measurement outcomes to measure the observable whose index maximizes the UCB on the baseline increment value.

In Figure \ref{fig:multiple_observables} we report the distribution of the detection delay for these schemes. In particular, we fix the changepoint location at $\nu=200$ and vary the number of observables defining the changepoint from $n=1$ to $n=10$. 
For a small number of observable $n$, eMCD schemes benefit from having their measurements tailored to the specific observables, while eSCD is penalized by its universality requirement. For this reason, for a number of observable $n\leq5$, eSCD has the largest detection delay. 

However, the cost of universality is rapidly offset as the number of observables, $n$, increases. 
In fact, as $n$ increases, eMCD schemes must allocate their measurement budget across a larger set of observables. For eMCD, the allocation is non-adaptive and fixed, whereas for adaptive eMCD it is based on past outcomes. As a result, the performance decrease is more pronounced for eMCD; nonetheless, both detectors eventually fall behind eSCD.  Specifically, eSCD becomes more efficient than eMCD for $n>5$, and surpasses adaptive eMCD for $n>9$. As $n$ increases, eSDC maintains nearly constant performance while the performance of both eMCD schemes keeps decreasing.
\subsection{Local vs.\ Joint Clifford Measurements}
\label{sec:local_vs_joint}

\begin{figure}
	\centering
	\includegraphics[width=0.95\textwidth]{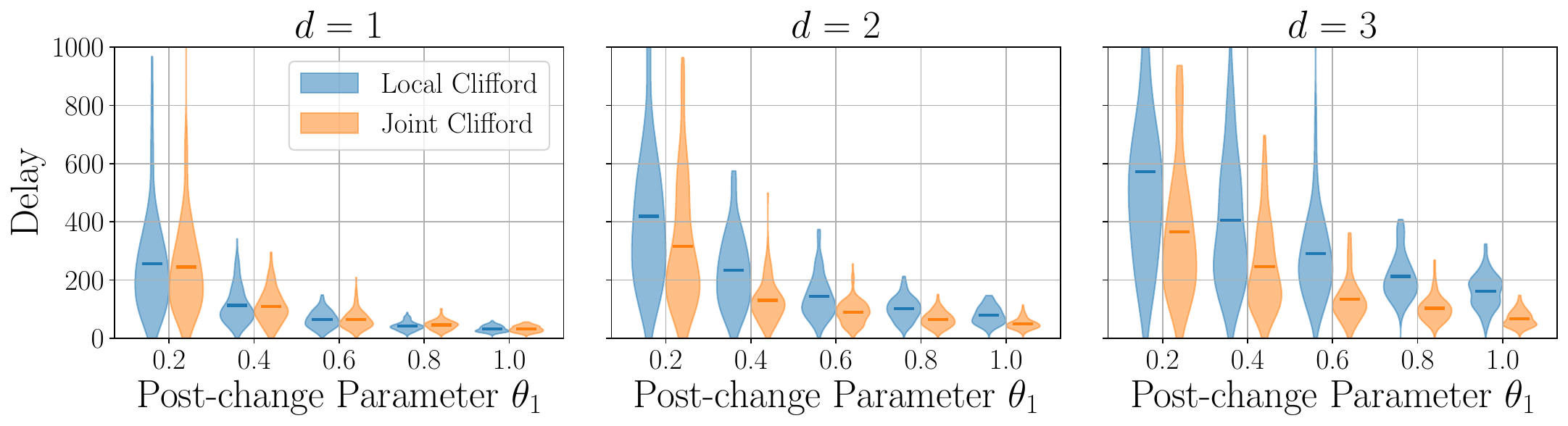}
	\caption{Violin plots for distribution of the detection delay as a function of the post-change parameter $\theta_1$ and of the number of qubits $d$ for eSCD with local Clifford ensemble and joint Clifford ensemble. Quantities are averaged over 100 runs, and all schemes target the ARL constraint $1/\alpha=1000$.}
	\label{fig:joint_vs_clifford}
\end{figure}

We now compare the efficiency of eSCD when implemented using either the local Clifford ensemble or the joint Clifford ensemble introduced in Section \ref{sec:classical_shadows}. From an implementation standpoint, local Clifford measurements are appealing as they require only single-qubit rotations. In contrast, sampling from the full $d$-qubit Clifford group generally requires deeper, entangling circuits. The advantage of joint Clifford measurements, however, is that they provide shadow estimates with smaller variance for a broad class of observables \cite{huang2020predicting}.

To illustrate the potential advantages of global measurements in the context of universal sequential changepoint detection, we revisit the single-observable changepoint experiment of Section \ref{sec:exp_single_obs} and evaluate eSCD under both strategies for randomized measurements. In Figure \ref{fig:joint_vs_clifford} we report the delay distributions for a number of qubits $d \in \{1,2,3\}$ and different values of the post-change parameter $\theta_1$. For $d = 1$, the performance of the two measurement schemes is nearly indistinguishable, and they both achieve comparable detection delays across all values of $\theta_1$. However, as the number of qubits increases, the joint Clifford ensemble exhibits a much more graceful degradation in detection performance, while the local Clifford ensemble suffers increasingly long detection delays. For example, when $d = 3$ and $\theta_1 \ge 0.6$, the average detection delay under local Clifford measurements is roughly twice that obtained under joint Clifford measurements.

\section{Conclusion}
\label{sec:conclusions}
In this work, we  have studied a novel sequential quantum changepoint detection problem, in which the pre- and post-change regimes are specified via constraints on the expectation values of a set of observables that may be unknown by the measurement device. For this setting, we proposed eSCD, a universal sequential changepoint detection algorithm that leverages a randomized measurement protocol and classical shadows \cite{huang2020predicting} for the universal measurement device, while using the framework of e-detectors to address nonparametric changepoint detection \cite{shin2022detectors}. Theoretically, eSCD is shown to enjoy nonasymptotic average run length control and efficiency guarantees on the worst-case detection delay. Experimentally, our results corroborate the theory and illustrate that eSCD can approach the performance of matched, observable-specific measurement strategies while maintaining universality.

Future research directions include the  extension of eSCD to support the localization of changepoints \cite{verzelen2023optimal}, as well as  the application of multi-stream changepoint detection \cite{dandapanthula2025multiple}  to identify which observables have changed, while guaranteeing constraints on false discovery rates.

\section*{Acknowledgments}
The work of O. Simeone was supported by the European Research Council (ERC) under the European Union’s Horizon Europe Programme (grant agreement No. 101198347), by an Open Fellowship of the EPSRC (EP/W024101/1), and by the EPSRC project (EP/X011852/1).

\printbibliography

@article{georgescu2014quantum,
  title={Quantum simulation},
  author={Georgescu, Iulia M and Ashhab, Sahel and Nori, Franco},
  journal={Reviews of Modern Physics},
  volume={86},
  number={1},
  pages={153--185},
  year={2014},
  publisher={APS}
}

@book{simeone2026classical, author    = {Osvaldo Simeone}, title     = {Classical and Quantum Information Theory}, publisher = {Cambridge University Press}, year      = {2026},  address   = {Cambridge, UK}}

@article{simeone2022introduction,
  title={An introduction to quantum machine learning for engineers},
  author={Simeone, Osvaldo},
  journal={Foundations and Trends{\textregistered} in Signal Processing},
  volume={16},
  number={1-2},
  pages={1--223},
  year={2022},
  publisher={Now Publishers, Inc.}
}

@article{arnon2018practical,
  title={Practical device-independent quantum cryptography via entropy accumulation},
  author={Arnon-Friedman, Rotem and Dupuis, Fr{\'e}d{\'e}ric and Fawzi, Omar and Renner, Renato and Vidick, Thomas},
  journal={Nature Communications},
  volume={9},
  number={1},
  pages={459},
  year={2018},
  publisher={Nature Publishing Group UK London}
}

@article{hensen2015loophole,
  title={Loophole-free Bell inequality violation using electron spins separated by 1.3 kilometres},
  author={Hensen, Bas and Bernien, Hannes and Dr{\'e}au, Ana{\"\i}s E and Reiserer, Andreas and Kalb, Norbert and Blok, Machiel S and Ruitenberg, Just and Vermeulen, Raymond FL and Schouten, Raymond N and Abell{\'a}n, Carlos and others},
  journal={Nature},
  volume={526},
  number={7575},
  pages={682--686},
  year={2015},
  publisher={Nature Publishing Group UK London}
}

@article{hayashi2001asymptotics,
	title={Asymptotics of quantum relative entropy from a representation theoretical viewpoint},
	author={Hayashi, Masahito},
	journal={Journal of Physics A: Mathematical and General},
	volume={34},
	number={16},
	pages={3413},
	year={2001},
	publisher={IOP Publishing}
}

@article{zecchin2025quantum,
	title={Quantum sequential universal hypothesis testing},
	author={Zecchin, Matteo and Simeone, Osvaldo and Ramdas, Aaditya},
	journal={arXiv preprint arXiv:2508.21594},
	year={2025}
}

@article{fanizza2023ultimate,
	title={Ultimate limits for quickest quantum change-point detection},
	author={Fanizza, Marco and Hirche, Christoph and Calsamiglia, John},
	journal={Physical review letters},
	volume={131},
	number={2},
	pages={020602},
	year={2023},
	publisher={APS}
}

@article{lorden1971procedures,
	title={Procedures for reacting to a change in distribution},
	author={Lorden, Gary},
	journal={The Annals of Mathematical Statistics},
	pages={1897--1908},
	year={1971},
	publisher={JSTOR}
}

@article{shin2022detectors,
	title={E-detectors: A nonparametric framework for sequential change detection},
	author={Shin, Jaehyeok and Ramdas, Aaditya and Rinaldo, Alessandro},
	journal={New England Journal on Statistics in Data Science},
	year={2023}
}

@article{page1954continuous,
	title={Continuous inspection schemes},
	author={Page, Ewan S},
	journal={Biometrika},
	volume={41},
	number={1/2},
	pages={100--115},
	year={1954},
	publisher={JSTOR}
}

@article{huang2020predicting,
	title={Predicting many properties of a quantum system from very few measurements},
	author={Huang, Hsin-Yuan and Kueng, Richard and Preskill, John},
	journal={Nature Physics},
	volume={16},
	number={10},
	pages={1050--1057},
	year={2020},
	publisher={Nature Publishing Group UK London}
}

@article{waudby2025universal,
	title={Universal log-optimality for general classes of e-processes and sequential hypothesis tests},
	author={Waudby-Smith, Ian and Sandoval, Ricardo and Jordan, Michael I},
	journal={arXiv preprint arXiv:2504.02818},
	year={2025}
}

@article{park2025adaptive,
	title={Adaptive prediction-powered AutoEval with reliability and efficiency guarantees},
	author={Park, Sangwoo and Zecchin, Matteo and Simeone, Osvaldo},
	journal={arXiv preprint arXiv:2505.18659},
	year={2025}
}

@inproceedings{jun2017improved,
	title={Improved strongly adaptive online learning using coin betting},
	author={Jun, Kwang-Sung and Orabona, Francesco and Wright, Stephen and Willett, Rebecca},
	booktitle={Artificial Intelligence and Statistics},
	pages={943--951},
	year={2017},
	organization={PMLR}
}

@article{cover1991universal,
	title={Universal portfolios},
	author={Cover, Thomas M},
	journal={Mathematical Finance},
	volume={1},
	number={1},
	pages={1--29},
	year={1991},
	publisher={Wiley Online Library}
}

@article{akimoto2011discrimination,
	title={Discrimination of the change point in a quantum setting},
	author={Akimoto, Daiki and Hayashi, Masahito},
	journal={Physical Review A—Atomic, Molecular, and Optical Physics},
	volume={83},
	number={5},
	pages={052328},
	year={2011},
	publisher={APS}
}

@ARTICLE{10220229,
	author={Shekhar, Shubhanshu and Ramdas, Aaditya},
	journal={IEEE Transactions on Information Theory}, 
	title={Nonparametric two-sample testing by betting}, 
	year={2024},
	volume={70},
	number={2},
	pages={1178-1203},
	doi={10.1109/TIT.2023.3305867}}

@article{agrawal2025stopping,
		title={On stopping times of power-one sequential tests: tight lower and upper bounds},
		author={Agrawal, Shubhada and Ramdas, Aaditya},
		journal={arXiv preprint arXiv:2504.19952},
		year={2025}
	}

@ARTICLE{10315047,
		author={Orabona, Francesco and Jun, Kwang-Sung},
		journal={IEEE Transactions on Information Theory}, 
		title={Tight concentrations and confidence sequences from the regret of universal portfolio}, 
		year={2024},
		volume={70},
		number={1},
		pages={436-455},
		doi={10.1109/TIT.2023.3330187}}

@article{RAMDAS202283,
			title = {Testing exchangeability: Fork-convexity, supermartingales and e-processes},
			journal = {International Journal of Approximate Reasoning},
			volume = {141},
			pages = {83-109},
			year = {2022},
			note = {Probability and Statistics: Foundations and History. In honor of Glenn Shafer},
			issn = {0888-613X},
			doi = {https://doi.org/10.1016/j.ijar.2021.06.017},
			url = {https://www.sciencedirect.com/science/article/pii/S0888613X21000980},
			author = {Aaditya Ramdas and Johannes Ruf and Martin Larsson and Wouter M. Koolen},
		}

@article{casgrain2024sequential,
			title={Sequential testing for elicitable functionals via supermartingales},
			author={Casgrain, Philippe and Larsson, Martin and Ziegel, Johanna},
			journal={Bernoulli},
			volume={30},
			number={2},
			pages={1347--1374},
			year={2024},
			publisher={Bernoulli Society for Mathematical Statistics and Probability}
		}

@article{sentis2018online,
			title={Online strategies for exactly identifying a quantum change point},
			author={Sent{\'\i}s, Gael and Mart{\'\i}nez-Vargas, Esteban and Munoz-Tapia, Ramon},
			journal={Physical Review A},
			volume={98},
			number={5},
			pages={052305},
			year={2018},
			publisher={APS}
		}

@article{nakahira2025unambiguous,
			title={Unambiguous discrimination of the change point for quantum channels},
			author={Nakahira, Kenji},
			journal={arXiv preprint arXiv:2508.06785},
			year={2025}
		}

@article{john2025fundamental,
			title={Fundamental limits Of quickest change-point detection with continuous-variable quantum states},
			author={John, Tiju Cherian and Gagatsos, Christos N and Bash, Boulat A},
			journal={arXiv preprint arXiv:2504.16259},
			year={2025}
		}

@article{sentis2017exact,
			title={Exact identification of a quantum change point.},
			author={Sent{\'\i}s, G and Calsamiglia, J and Mu{\~n}oz-Tapia, R},
			journal={Physical Review Letters},
			volume={119},
			number={14},
			pages={140506--140506},
			year={2017}
		}

@article{sentis2016quantum,
			title={Quantum change point},
			author={Sentis, G and Bagan, E and Calsamiglia, J and Chiribella, G and Mu{\~n}oz-Tapia, R},
			journal={Physical Review Letters},
			year={2016},
			publisher={American Physical Society. The Journal's web site is located at http://prl~…}
		}

@inproceedings{shiryaev1961problem,
			title={The problem of the most rapid detection of a disturbance in a stationary process},
			author={Shiryaev, Albert N},
			booktitle={Soviet Math. Dokl},
			volume={2},
			number={795-799},
			pages={103},
			year={1961}
		}

@article{pollak1985optimal,
			title={Optimal detection of a change in distribution},
			author={Pollak, Moshe},
			journal={The Annals of Statistics},
			pages={206--227},
			year={1985},
			publisher={JSTOR}
		}

@article{koenig2014efficiently,
			title={How to efficiently select an arbitrary Clifford group element},
			author={Koenig, Robert and Smolin, John A},
			journal={Journal of Mathematical Physics},
			volume={55},
			number={12},
			pages={122202},
			year={2014},
			publisher={AIP Publishing LLC}
		}

@inproceedings{hazan2007adaptive,
			title={Adaptive algorithms for online decision problems},
			author={Hazan, Elad and Seshadhri, Comandur},
			booktitle={Electronic Colloquium on Computational Complexity (ECCC)},
			volume={14},
			number={088},
			year={2007}
		}

@inproceedings{hazan2009efficient,
			title={Efficient learning algorithms for changing environments},
			author={Hazan, Elad and Seshadhri, Comandur},
			booktitle={Proceedings of the 26th Annual International Conference on Machine Learning},
			pages={393--400},
			year={2009}
		}

@inproceedings{daniely2015strongly,
			title={Strongly adaptive online learning},
			author={Daniely, Amit and Gonen, Alon and Shalev-Shwartz, Shai},
			booktitle={International Conference on Machine Learning},
			pages={1405--1411},
			year={2015},
			organization={PMLR}
		}

@inproceedings{wang2018minimizing,
			title={Minimizing adaptive regret with one gradient per iteration.},
			author={Wang, Guanghui and Zhao, Dakuan and Zhang, Lijun},
			booktitle={IJCAI},
			pages={2762--2768},
			year={2018}
		}

@inproceedings{cutkosky2020parameter,
			title={Parameter-free, dynamic, and strongly-adaptive online learning},
			author={Cutkosky, Ashok},
			booktitle={International Conference on Machine Learning},
			pages={2250--2259},
			year={2020},
			organization={PMLR}
		}

@article{terhal2000bell,
  title={Bell inequalities and the separability criterion},
  author={Terhal, Barbara M},
  journal={Physics Letters A},
  volume={271},
  number={5-6},
  pages={319--326},
  year={2000},
  publisher={Elsevier}
}

@article{grootveld2026asymptotically,
  title={Asymptotically optimal quantum universal quickest change detection},
  author={Grootveld, Arick and Yang, Haodong and Sriranga, Nandan and Chen, Biao and Gandikota, Venkata and Pollack, Jason},
  journal={arXiv preprint arXiv:2602.02950},
  year={2026}
}

@article{cumitini2026anytime,
  title={Anytime-valid quantum tomography via confidence sequences},
  author={Cumitini, Aldo and Barletta, Luca and Simeone, Osvaldo},
  journal={arXiv preprint arXiv:2601.20761},
  year={2026}
}

@article{verzelen2023optimal,
  title={Optimal change-point detection and localization},
  author={Verzelen, Nicolas and Fromont, Magalie and Lerasle, Matthieu and Reynaud-Bouret, Patricia},
  journal={The Annals of Statistics},
  volume={51},
  number={4},
  pages={1586--1610},
  year={2023},
  publisher={Institute of Mathematical Statistics}
}

@article{dandapanthula2025multiple,
  title={Multiple testing in multi-stream sequential change detection},
  author={Dandapanthula, Sanjit and Ramdas, Aaditya},
  journal={arXiv preprint arXiv:2501.04130},
  year={2025}
}

@incollection{feynman2018simulating,
  title={Simulating physics with computers},
  author={Feynman, Richard P},
  booktitle={Feynman and Computation},
  pages={133--153},
  year={2018},
  publisher={cRc Press}
}

@article{cerezo2021variational,
  title={Variational quantum algorithms},
  author={Cerezo, Marco and Arrasmith, Andrew and Babbush, Ryan and Benjamin, Simon C and Endo, Suguru and Fujii, Keisuke and McClean, Jarrod R and Mitarai, Kosuke and Yuan, Xiao and Cincio, Lukasz and others},
  journal={Nature Reviews Physics},
  volume={3},
  number={9},
  pages={625--644},
  year={2021},
  publisher={Nature Publishing Group UK London}
}

@book{gottesman1997stabilizer,
  title={Stabilizer codes and quantum error correction},
  author={Gottesman, Daniel},
  year={1997},
  publisher={California Institute of Technology}
}

@article{degen2017quantum,
  title={Quantum sensing},
  author={Degen, Christian L and Reinhard, Friedemann and Cappellaro, Paola},
  journal={Reviews of Modern Physics},
  volume={89},
  number={3},
  pages={035002},
  year={2017},
  publisher={APS}
}

@inproceedings{aaronson2018shadow,
  title={Shadow tomography of quantum states},
  author={Aaronson, Scott},
  booktitle={Proceedings of the 50th Annual ACM SIGACT Symposium on Theory of Computing},
  pages={325--338},
  year={2018}
}

@article{ramdas2025hypothesis,
  title={Hypothesis testing with e-values},
  author={Ramdas, Aaditya and Wang, Ruodu},
  journal={Foundations and Trends in Statistics},
  volume={1},
  number={1-2},
  pages={1--390},
  year={2025},
  publisher={Emerald Publishing Limited}
}

@article{ram2026asymptotically,
  title={Asymptotically optimal sequential change detection for bounded means},
  author={Ram, Ashwin and Ramdas, Aaditya},
  journal={arXiv preprint arXiv:2602.05272},
  year={2026}
}
\appendix
\section{Proof of Proposition \ref{prop:bound_o_hat_general}}
\label{app:bound_o_hat_general}
\begin{proof}[Proof of Proposition \ref{prop:bound_o_hat_general}]
	The unbiasedness of $\hat{o}^t$ is immediate from the unbiasedness of classical shadows,
	\begin{align}
	\mathbb{E}_{U^t,X^t}[\hat{\rho}^t] = \rho^t.
	\end{align}
	It follows that, for any observable $O$, it holds
	\begin{align}
	\mathbb{E}_{U^t,X^t}[\hat{o}^t] = \mathbb{E}_{U^t,X^t}\big[\Tr(O\hat{\rho}^t)\big]
	= \Tr\big(O\mathbb{E}_{U^t,X^t}[\hat{\rho}^t]\big)
	= \Tr(O\rho)
	= \langle O\rangle_\rho.
	\end{align}
	
	We now proceed to show the almost-sure boundedness of $\hat{o}^t$ for Local Clifford and Joint Clifford ensembles separately. 
	
	\emph{Local Clifford ensemble:}  Let $S\subseteq\{1,\dots,d\}$ denote the subset of qubits on which $O$ acts non-trivially.
	Since $O$ acts trivially outside $S$, writing $O = O_{S} \otimes I_{S^c}$, it follows
	\begin{align}
		\hat{o}^t = \Tr(O\hat\rho^t) = \Tr\big( O_{i,{S}}  \hat\rho^t_{S}\big),
	\end{align}
	where $\hat\rho^t_{S}$ is the snapshot $\hat\rho^t$ restricted on qubits in $S$,  i.e.,
	\begin{align}
		\hat\rho^t_{S} = \bigotimes_{k\in S} 3 (U^t)^{\dag}_k\ket{X^t_k}\bra{X^t_k}U^t_k - I
	\end{align}
	Each factor,
	\begin{align}
		3 (U^t)^{\dag}_k\ket{X^t_k}\bra{X^t_k}U^t_k - I
	\end{align}
	is Hermitian with eigenvalues $2$ and $-1$, hence its trace norm is $3$. It follows that
	\begin{align}
		\|\hat\rho_{S}\|_1
		= \prod_{k\in S} \|3 (U^t)^{\dag}_k\ket{X^t_k}\bra{X^t_k}U^t_k - I\|_1
		= 3^{|S|}.
	\end{align}
	From Hölder's inequality, it follows  
	\begin{equation}
		\label{eq:trace-bound-general}
			\bigl|	\hat{o}^t \bigr|=\bigl| \Tr(O_{i,{S}}  \hat\rho^t_{S})\bigr|
		\le \|O_{i,{S}}\|_\infty \| \hat\rho^t_{S}\|_1\leq 3^{|S|} \|O\|_\infty.
	\end{equation}
	We conclude that for any observable $O$ acting non-trivially on $|S|$ qubits and with $\|O\|_\infty<\infty$, $\hat{o}^t$ admits finite upper and lower bound $l=-3^{|S|} \|O\|_\infty$ and $u=3^{|S|} \|O\|_\infty$. Note that for a small number of qubits $d$ one can obtain tight upper and lower bounds by evaluating the inverse channel map \eqref{eq:Pauli-inverse-map} for all combinations $U$ and $X$.

	\emph{Joint Clifford ensemble:}
	If $U$ is drawn from the full $d$-qubit Clifford group,
	then every snapshot has the form
	\begin{align}
	\hat{\rho}=(2^{d}+1)\ket{\psi}\bra{\psi}-I,
	\end{align}
	with some pure state $\ket\psi\in\mathbb C^{2^d}$. Consequently, writing
	$\lambda_{\max}(O)$ and $\lambda_{\min}(O)$ for the largest and smallest eigenvalues of $O$,
	the estimator $\hat{o}^t=\Tr(O\hat\rho^t)$ satisfies the bounds
	\begin{align}
	l=(2^{d}+1)\lambda_{\min}(O)-\Tr(O)
	\le
	\hat{o}^t
	\le
	(2^{d}+1)\lambda_{\max}(O)-\Tr(O)=u
	\end{align}
	where both $l$ and $u$ are bounded having assumed $\lVert O\rVert_\infty<\infty$.
\end{proof}

\section{eSCD based on CUSUM e-detectors}
\label{sec:cusum_detector}
As an alternative to the eSCD based on the SR construction presented in Section \ref{sec:e-SB-SR-detector}, it is possible to obtain an eSCD based on CUSUM e-detectors (eSCD-CU) by aggregating the baseline increment process \eqref{eq:baseline_increment} using the CUSUM rule \cite{page1954continuous}. Specifically, the corresponding per-observable e-detector is defined as
\begin{align}
	\label{eq:e-detectors_CUSUM}
	M^t_{\rm{CU},i}&=L_i^t\cdot\max\left\{M^{t-1}_{\rm{CU},i},1\right\}=\max_{j\in[t]}\prod^t_{k=j}L^j_i=\max_{j\in[t]}E^{j:t}.
\end{align}

Given the per-observable e-detectors $\{M^t_{\rm{CU},i}\}^n_{i=1}$ the aggregate CUSUM e-detector for the set $\mathcal{S}_0$ is given by
\begin{align}
	\label{eq:mixture_e-detector_CUSUM}
	M^t_{\rm{CU}}&=\sum^n_{i=1}w_i	M^t_{\rm{CU},i}=\sum^n_{i=1}w_i	\max_{j\in[t]}E^{j:t}_i.
\end{align}
Thresholding the value of the CUSUM e-detector,  the eSCD-CU decision rule is
\begin{align}
	\phi^t(M^t_{\rm{CU}})=D^t=\begin{cases}
		1, \quad \text{ if } 	M^t_{\rm{CU}}>c_\alpha,\\
		0, \quad \text{ otherwise.}
	\end{cases}
\end{align}
This decision rule induces the stopping times
\begin{align}
	\label{eq:stopping_times_aggr_2}
	T_{\rm{eSCD-CU}}&=\inf\{t\geq1: M^t_{\rm{CU}}\geq  c_\alpha\}.
\end{align}
The threshold can $c_\alpha$ be conservatively set to $1/\alpha$. The eSCD-CU enjoys the same ARL and efficiency guarantees as those stated in Theorems \ref{th:arl_control} and \ref{th:eff_theorem} for the e-SB-SR detector; however, it is empirically less efficient than its SR counterpart in our experiments. Nonetheless, CUSUM detectors are known to possess strong minimax-type optimality guarantees for quickest change detection \cite{pollak1985optimal}.
\section{Proof of Theorem \ref{th:arl_control}}
\label{proof:arl_control}
We prove a generalized version Theorem \ref{th:arl_control} by lower bounding  eSCD's $T$ under any sequence of state $\{\rho^t_0\}_{t\geq1}$ such that $\rho^t_0\in\mathcal{S}_0$ for all $t\geq1$. The statement of Theorem \ref{th:arl_control} setting $\rho^t_0=\rho_0$. We define the expectation operator with respect to the sequence of quantum states $\{\rho^t_0\}_{t\geq1}$ as $\mathbb{E}_{\{\rho^t_0\}_{t\geq1}}\left[\cdot\right]$.

The proof is based on the property of e-detectors and the unbiasedness of the estimates $\{\{\hat{o}_i^t\}_{t\geq 1}\}^n_{i=1}$. A process $M^t$ is a valid e-detector; i.e., for any stopping time $\tau$ and sequence $\{\rho^t_0\}_{t\geq1}$ it holds
\begin{align}
    \mathbb{E}_{\{\rho^t_0\}_{t\geq1}}\left[M^\tau\right]\leq  \mathbb{E}_{\{\rho^t_0\}_{t\geq1}}\left[\tau\right].
\end{align}
Both the SR e-detector  $M_{\rm{SR}}^t$ in \eqref{eq:e-detector} and the CUSUM e-detector $M^t_{\rm{CU}}$ \eqref{eq:mixture_e-detector_CUSUM} are valid e-detectors, in fact \cite{shin2022detectors}
\begin{align}
    \mathbb{E}_{\{\rho^t_0\}_{t\geq1}}\left[M_{\rm{CU}}^\tau\right]\leq\mathbb{E}_{\{\rho^t_0\}_{t\geq1}}\left[M_{\rm{SR}}^\tau\right]&=\sum^n_{i=1}w_i\mathbb{E}_{\{\rho^t_0\}_{t\geq1}}\left[\sum^\infty_{j=1}\mathds{1}\left\{j\leq\tau\right\} E_i^{j:\tau}\right]\\
    &=\sum^n_{i=1}w_i\sum^\infty_{j=1}\mathbb{E}_{\{\rho^t_0\}^{j-1}_{t=1}}\left[\mathds{1}\left\{j\leq\tau\right\}\mathbb{E}_{\{\rho^t_0\}_{t\geq j}}\left[E_i^{j:\tau}\big|H^{j-1}\right]\right]\\
    &\leq\sum^n_{i=1}w_i\sum^\infty_{j=1}\mathbb{E}_{\{\rho^t_0\}^{j-1}_{t=1}}\left[\mathds{1}\left\{j\leq\tau\right\}\right]=\mathbb{E}_{\{\rho^t_0\}_{t\geq1}}\left[\tau\right].
    \label{eq:e-detector-property}
\end{align}
where the first equality follows from the non-negativity of the e-detector, the second from the law of total expectation, while the last inequality follows from the definition of baseline increments, predictability of the betting strategy, and the unbiasedness of the estimates$\{\{\hat{o}_i^t\}_{t\geq 1}\}^n_{i=1}$.
\begin{align}
    \mathbb{E}_{\{\rho^t_0\}^{\tau}_{t=j}}\left[E_i^{j:\tau}\big|H^{j-1}\right]&=\mathbb{E}_{\{\rho^t_0\}^{\tau}_{t=j}}\left[\prod^\tau_{t=j} L^t_i\Bigg|H^{j-1}\right]\\
    &=\mathbb{E}_{\{\rho^t_0\}^{\tau-1}_{t=j}}\left[\prod^{\tau-1}_{t=j} L^t_i\mathbb{E}[(1+\lambda^\tau_i \hat{o}^\tau_i)|H^{\tau-1}]\Bigg|H^{j-1}\right]\\
    &\leq\mathbb{E}_{\{\rho^t_0\}^{\tau-1}_{t=j}}\left[\prod^{\tau-1}_{t=j} L^t_i\Bigg|H^{j-1}\right]\leq...\leq1.
\end{align}
It then follows that for the stopping time \eqref{eq:stopping_times_aggr_1} and \eqref{eq:stopping_times_aggr_2} (with $c_\alpha=1/\alpha$), the e-detectors satisfy $M_{\rm{SR}}^T\geq 1/\alpha$ and $M^T_{\rm{CU}}\geq 1/\alpha$. This fact, combined with inequality \eqref{eq:e-detector-property}, implies that when the stopping time $T$ is finite, 
\begin{align}
     \mathbb{E}_{\{\rho^t_0\}_{t\geq1}}\left[T_{\rm{eSCD}}\right]\geq\mathbb{E}_{\{\rho^t_0\}_{t\geq1}}\left[M_{\rm{SR}}^T\right]\geq 1/\alpha,\\
     \mathbb{E}_{\{\rho^t_0\}_{t\geq1}}\left[T_{\rm{eSCD-CU}}\right]\geq\mathbb{E}_{\{\rho^t_0\}_{t\geq1}}\left[M_{\rm{CU}}^T\right]\geq 1/\alpha,
\end{align}
for any sequence $\{\rho^t_0\}_{t\geq1}$.
\section{Proof of Theorem \ref{th:eff_theorem}}
\label{proof:eff_theorem}

Following the proof strategy in \cite{park2025adaptive}, we base the analysis of the expected detection delay on the quantity
\begin{align}
	\label{eq:surrogate_prod_incr}
	\tilde{E}^{k:t} = \sum_{i=1}^n w_i \min\{E^{k:t}_i, E^{k:t}_{i^*}\},
\end{align}
which serves as a lower bound to the mixture of baseline increments' partial products
\begin{align}
	\tilde{E}^{k:t} = \sum_{i=1}^n w_i \min\{E^{k:t}_i, E^{k:t}_{i^*}\} \leq \sum_{i=1}^n w_i E^{k:t}_i.
\end{align}
Lower bounds on the stopping time of eSCD and eSCD-CU can be obtained by studying the quantity \eqref{eq:surrogate_prod_incr} since for eSCD's stopping time, it holds
\begin{align}
	T_{\rm{eSCD}}&=\inf\left\{t\geq 1:\sum^t_{j=1}\sum^m_{i=1}w_i	E^{j:t}_i\geq 1/\alpha\right\}\nonumber\\
	&\leq\inf\left\{t\geq 1:\sum^t_{j=1}\tilde{E}^{j:t}\geq 1/\alpha\right\}\nonumber\\
	&\leq \min_{j\geq 1}\inf\left\{t\geq j:\tilde{E}^{j:t}\geq 1/\alpha\right\}:=\min_{j\geq1}N^j,
\end{align}
while for eSCD-CU's stopping time, it holds
\begin{align}
	T_{\rm{eSCD-CU}}&=\inf\left\{t\geq 1:\sum^{m}_{i=1}w_i\max_{j\in[t]}E_i^{j:t}\geq c_\alpha\right\}\nonumber\\
	&\leq\inf\left\{t\geq 1:\max_{j\in[t]}\sum^{m}_{i=1}w_iE_i^{j:t}\geq c_\alpha\right\}\nonumber\\
	&\leq\inf\left\{t\geq 1:\max_{j\in[t]}\tilde{E}^{j:t}\geq c_\alpha\right\}\nonumber\\
	&\leq \min_{j\geq 1}\inf\left\{t\geq j:\tilde{E}^{j:t}\geq c_\alpha\right\}:=\min_{j\geq1}N^j.
\end{align}
In the following, we focus on eSCD, since the proof for eSCD-CU follows by replacing $1/\alpha$ with $c_\alpha$. The following proof technique is inspired by results from \cite{shin2022detectors,waudby2025universal,park2025adaptive}.

Since we are interested in upper bounding the worst-case delay $\tau^*(\rho_1)$ defined in \eqref{eq:worst_worst_case_delay}, in the following we condition on the pre-change quantum state pair $\rho_0\in\mathcal{S}_0$,  changepoint location $\nu\geq 1$ and sequence of measurements outcomes and POVMs $H^\nu$ up to time $\nu$. For any such conditioning, the average detection delay of eSCD can be bounded as
\begin{align}
	\mathbb{E}_{\rho_0,\nu,\rho_1}\left[[T_{\rm eSCD}-\nu]_+\big| H^\nu\right]&\leq\mathbb{E}_{\rho_0,\nu,\rho_1}\left[\left[ \min_{j\geq 1}N^{j}-\nu\right]_+\big| H^\nu\right]\nonumber\\
	&\leq \min_{j\geq \nu}\mathbb{E}_{\rho_0,\nu,\rho_1}\left[N^{j}\big| H^\nu\right]-\nu\nonumber\\
	&\leq \mathbb{E}_{\rho_0,\nu,\rho_1}\left[N^{\nu}\big| H^\nu\right]-\nu
\end{align}

As in \cite{waudby2025universal}, for a given $\delta\in(0,1)$ we define the two following auxiliary quantities
\begin{align}
	\epsilon:=\frac{\delta}{1+\delta}D^*(\rho_1), \qquad m:=\left\lceil\frac{\log(1/\alpha)}{D^*(\rho_1)-\epsilon}\right\rceil=\left\lceil\frac{(1+\delta)\log(1/\alpha)}{D^*(\rho_1)}\right\rceil
\end{align}
Denoting $m'=m+\nu$ it follows that
\begin{align}
	\mathbb{E}_{\rho_0,\nu,\rho_1}\left[N^{\nu}\big| H^\nu\right]&=\mathbb{E}_{\rho_0,\nu,\rho_1}\left[N^{\nu}\mathds{1}\{N^{\nu}\leq m'\}\big| H^\nu\right]+\mathbb{E}_{\rho_0,\nu,\rho_1}\left[N^{\nu}\mathds{1}\{N^{\nu}>m'\}\big| H^\nu\right]\nonumber\\
	&\leq m'\mathbb{P}_{\rho_0,\nu,\rho_1}\left[N^{\nu}\leq m'\big| H^\nu\right]+\sum^\infty_{k=m'+1}k\mathbb{P}_{\rho_0,\nu,\rho_1}\left[N^{\nu}=k\big| H^\nu\right]\nonumber\\
	&\leq m'\mathbb{P}_{\rho_0,\nu,\rho_1}\left[N^{\nu}\leq m'\big| H^\nu\right]+\sum^\infty_{k=1}(k+m')\mathbb{P}_{\rho_0,\nu,\rho_1}\left[N^{\nu}=k+m'\big| H^\nu\right]\nonumber\\
	&\leq m'+\sum^\infty_{k=1}k\mathbb{P}_{\rho_0,\nu,\rho_1}\left[N^{\nu}=k+m'\big| H^\nu\right]\nonumber\\
	&\leq m'+\sum^\infty_{k=1}\mathbb{P}_{\rho_0,\nu,\rho_1}\left[N^{\nu}\geq k+m'\big| H^\nu\right]\nonumber\\
	&\leq m'+\underbrace{\sum^\infty_{k=m'}\mathbb{P}_{\rho_0,\nu,\rho_1}\left[\tilde{E}^{\nu:k}< 1/\alpha \Big| H^\nu\right]}_{(*)}.
\end{align}
Note that
\begin{align}
	\tilde{E}^{\nu:k}< 1/\alpha  \Longleftrightarrow \frac{\log \tilde{E}^{\nu:k}}{k-\nu}<\frac{\log(1/\alpha)}{k-\nu},
\end{align}
and because $m(D^*(\rho_1)-\epsilon)\geq \log(1/\alpha)$ and $k\geq m'=m+\nu$ 
\begin{align}
	\frac{\log \tilde{E}^{j:k}}{k-\nu}<\frac{\log(1/\alpha)}{k-\nu}\implies \frac{\log \tilde{E}^{\nu:k}}{k-\nu}<\frac{m(D^*(\rho_1)-\epsilon)}{k-\nu}\implies \frac{\log \tilde{E}^{\nu:k}}{k-\nu}<D^*(\rho_1)-\epsilon.
\end{align}
Based on the above, the term $(*)$ can be upper bounded as
\begin{align}
	(*)=&\sum^\infty_{k=m+\nu}\mathbb{P}_{\rho_0,\nu,\rho_1}\left[\tilde{E}^{\nu:k}< 1/\alpha \Big| H^\nu\right]\nonumber\\
	\leq& \sum^\infty_{k=m+\nu}\mathbb{P}_{\rho_0,\nu,\rho_1}\left[\left|\frac{1}{k-\nu}\log \tilde{E}^{\nu:k}-D^*(\rho_1)\right|\geq \epsilon \ \Bigg| H^\nu\right]\nonumber\\
	\leq& \sum^\infty_{k=m+\nu}\underbrace{\mathbb{P}_{\rho_0,\nu,\rho_1}\left[\left|\frac{1}{k-\nu}\left(\log \tilde{E}^{\nu:k}-\log E_{i^*}^{\nu:k}(\lambda^*_{i^*})\right)\right|\geq \frac{\epsilon}{2}\  \Bigg| H^\nu\right]}_{(\diamond)}\nonumber\\
	&+\underbrace{\sum^\infty_{k=m+\nu}\mathbb{P}_{\rho_0,\nu,\rho_1}\left[\left|\frac{1}{k-\nu}\log E_{i^*}^{\nu:k}(\lambda^*_{i^*})-D^*(\rho_1)\right|\geq \frac{\epsilon}{2} \ \Bigg| H^\nu\right]}_{(\ddag)}
	\label{eq:intermediate_1} 
\end{align}
For the  term $(\diamond)$ it holds
\begin{align}
	\diamond\leq& \mathbb{P}_{\rho_0,\nu,\rho_1}\left[\log\left(\frac{\tilde{E}^{\nu:k}}{E_{i^*}^{\nu:k}(\lambda^*_{i^*})}\right)\geq \frac{\epsilon(k-\nu)}{2}\ \Bigg| H^\nu\right]\nonumber\\
	&+\mathbb{P}_{\rho_0,\nu,\rho_1}\left[\frac{1}{k-\nu}\left(\log E_{i^*}^{\nu:k}(\lambda^*_{i^*})-\log \tilde{E}^{\nu:k}\right)\geq \frac{\epsilon}{2}\ \Bigg| H^\nu\right],
\end{align}
where the first term can be bounded as
\begin{align}
	\mathbb{P}_{\rho_0,\nu,\rho_1}\left[\log\left(\frac{\tilde{E}^{\nu:k}}{E_{i^*}^{\nu:k}(\lambda^*_{i^*})}\right)\geq \frac{\epsilon(k-\nu)}{2}\ \Bigg| H^\nu\right]&\leq \mathbb{E}_{\rho_0,\nu,\rho_1}\left[\frac{\tilde{E}^{\nu:k}}{E_{i^*}^{\nu:k}(\lambda^*_{i^*})} \Bigg| H^\nu\right]\exp\left\{-\frac{\epsilon(k-\nu)}{2}\right\}\nonumber\\
	&\leq \exp\left\{-\frac{\epsilon(k-\nu)}{2}\right\},
\end{align}
where the last inequality follows from \cite[Lemma 3]{park2025adaptive} given the numerarie property
\begin{align}
	\mathbb{E}_{\rho_0,\nu,\rho_1}\left[\frac{\tilde{E}^{\nu:k}}{E_{i^*}^{\nu:k}(\lambda^*_{i^*})} \Bigg| H^\nu\right]&=\mathbb{E}_{\rho_0,\nu,\rho_1}\left[\frac{\sum^n_{i=1}w_{i}\min\{E^{\nu:t}_i,E^{\nu:t}_{i^*}\}}{E_{i^*}^{\nu:k}(\lambda^*_{i^*})} \Bigg| H^\nu\right]\nonumber\\
	&\leq\mathbb{E}_{\rho_0,\nu,\rho_1}\left[\frac{E^{\nu:t}_{i^*}}{E_{i^*}^{\nu:k}(\lambda^*_{i^*})} \Bigg| H^\nu\right]\leq 1
\end{align}
The second  can be related to the strongly adaptive regret of the sequence of hyperparameters and bounded using the regret guarantee \eqref{eq:time_int_regret}
\begin{align}
	\Pr\left[\frac{\log E_{i^*}^{\nu:k}(\lambda^*_{i^*})-\log \tilde{E}^{\nu:k}}{k-\nu}\geq \frac{\epsilon}{2}\ \Bigg| H^\nu\right]&=\Pr\left[\mathcal{R}^{\nu:k}(\lambda_{i^*}^{\nu:k})-\mathcal{R}^{\nu:k}(\mathbf{w})\geq \frac{\epsilon(k-\nu)}{2}\ \Bigg| H^\nu\right]\nonumber\\
	&\leq \mathds{1}\left\{\mathcal{R}^{\nu:k}\geq \frac{\epsilon(k-\nu)}{2}\right\}.
\end{align}
It follows that the first sum in \eqref{eq:intermediate_1} can be bounded as
\begin{align}
	\sum^{\infty}_{k=m+\nu}\exp\left\{-\frac{\epsilon(k-\nu)}{2}\right\}&+\mathds{1}\left\{\mathcal{R}^{\nu:k}\geq \frac{\epsilon(k-\nu)}{2}\right\}\leq\nonumber\\
	&\leq 1+\frac{2}{\epsilon}\exp\{-m\epsilon/2\}+\mathds{1}\left\{\mathcal{R}^{\nu:k}\geq \frac{\epsilon(k-\nu)}{2}\right\}\nonumber\\
	&\leq 1+\frac{2(1+\delta)}{\delta D^*(\rho_1)}\exp\left\{-\frac{\delta \log(1/\alpha)}{2}\right\}+\mathds{1}\left\{\mathcal{R}^{\nu:k}\geq \frac{\epsilon(k-\nu)}{2}\right\}\nonumber\\
	&=1+\frac{2(1+\delta)\alpha^{\delta/2}}{\delta D^*(\rho_1)}+\mathds{1}\left\{\mathcal{R}^{\nu:k}\geq \frac{\epsilon(k-\nu)}{2}\right\}.
\end{align}

Finally, the term $(\ddag)$ can be bounded using Chebyshev-Nemirovski inequality \cite[Lemma B.1]{waudby2025universal} as
\begin{align}
	(\ddag)&\leq 1+\frac{2^s \sigma_s(\rho_1)}{\epsilon^{s/2}m^{s/2-1}(s/2-1)}=
	1+\frac{2^s \sigma_s(\rho_1)}{(s/2-1)}\left(\frac{1+\delta}{D^*(\rho_1)}+\frac{1}{\log(1/\alpha)}\right)
\end{align}
Combining all together, we get 
\begin{align}
    \label{eq:finite_alpha_bound}
	\mathbb{E}_{\rho_0,\nu,\rho_1}\left[[T_{\rm{eSCD}}-\nu]_+\big| H^\nu\right]\leq 2&+ \left\lceil\frac{(1+\delta)\log(1/\alpha)}{D^*(\rho_1)}\right\rceil+\sum^\infty_{k=m+\nu}\mathds{1}\left\{\mathcal{R}^{\nu:k}\geq \frac{\epsilon(k-\nu)}{2}\right\}\nonumber\\
	&+\frac{2^s \sigma_s(\rho_1)}{(s/2-1)}\left(\frac{1+\delta}{D^*(\rho_1)}+\frac{1}{\log(1/\alpha)}\right)+\frac{2(1+\delta)\alpha^{\delta/2}}{\delta D^*(\rho_1)}
\end{align}

Whenever the regret $\mathcal{R}^{\nu:k}$ admits a sublinear bound for all $k>\nu$, the asymptotic guarantee \eqref{eq:optimal_detection_delay} is obtained from \eqref{eq:finite_alpha_bound} for $\alpha\to0^+$. In this regime, the value of $m$ grows as $\log(1/\alpha)$ and the dominant term in \eqref{eq:finite_alpha_bound} becomes $\frac{(1+\delta)\log(1/\alpha)}{D^*(\rho_1)}$ since the terms in the series
\begin{align}
    \sum^\infty_{k=m+\nu}\mathds{1}\left\{\mathcal{R}^{\nu:k}\geq \frac{\epsilon(k-\nu)}{2}\right\}
\end{align} 
become zero for large $m$ under the sublinear upper bound assumption. Making $\delta\to1$ in the limit, one obtains \eqref{eq:optimal_detection_delay}.

We conclude by showing that the central moment of the log-increment satisfies
\begin{align}
\sigma_s(\rho_1)
=
\mathbb{E}_{\rho_1}\Bigl[\bigl|\log\bigl(E_{i^*}^{1}(\lambda^*_{i^*})\bigr)-D^*(\rho_1)\bigr|^s\Bigr]
<\infty .
\end{align}
Without loss of generality, assume that for each $i$ the observations $\{o_i^t\}_{t\ge1}$ are almost surely bounded in $[-1,1]$, and hence the admissible betting parameters satisfy $\lambda_i^t \in (-1,1)$.

Recall that $\lambda^*_{i^*}$ is the maximizer of
\begin{align}
f_D(\lambda) := \mathbb{E}_{\rho_1}\bigl[\log(1+\lambda \hat{o}_{i^*}^1)\bigr],
\end{align}
over $(0,1)$.
We show that if the distribution of $\hat{o}_{i^*}^1$ assigns positive probability mass arbitrarily close to $-1$, then $\lambda^*$ must be bounded away from $1$.

Fix $\delta>0$ and define the event
\begin{align}
E_\delta := \{\hat{o}_{i^*}^1 \in [-1,-1+\delta]\},
\qquad
p_\delta := \Pr(E_\delta).
\end{align}
Decomposing $f_D(\lambda)$ by conditioning yields
\begin{align}
f_D(\lambda)
&=
p_\delta \, \mathbb{E}_{\rho_1}\bigl[\log(1+\lambda \hat{o}_{i^*}^1)\mid E_\delta\bigr]
+
(1-p_\delta)\,\mathbb{E}_{\rho_1}\bigl[\log(1+\lambda \hat{o}_{i^*}^1)\mid \bar E_\delta\bigr].
\end{align}

Since $\log$ is concave and $\hat{o}_{i^*}^1 \le -1+\delta$ on $E_\delta$ and $\hat{o}_{i^*}^1 \le 1$ almost surely, we obtain for all $\lambda>0$
\begin{align}
f_D(\lambda)
\le
p_\delta \log\bigl(1+\lambda(-1+\delta)\bigr)
+
(1-p_\delta)\log(1+\lambda).
\end{align}
Taking the limit $\lambda\to1^-$ gives
\begin{align}
\limsup_{\lambda\to1^-} f_D(\lambda)
\le
p_\delta \log(\delta) + (1-p_\delta)\log(2).
\end{align}
If $\Pr(\hat{o}_{i^*}^1\in[-1,-1+\delta])>0$ for all $\delta>0$, then
\begin{align}
\lim_{\delta\to0}
\bigl[p_\delta \log(\delta) + (1-p_\delta)\log(2)\bigr]
=
-\infty,
\end{align}
which implies
\begin{align}
\lim_{\lambda\to1^-} f_D(\lambda) = -\infty.
\end{align}
Therefore, the maximizer $\lambda^*$ cannot lie arbitrarily close to $1$, and there exists $\varepsilon>0$ such that $\lambda^* \le 1-\varepsilon$.

Since $\lambda^*$ is bounded away from the boundary and $\hat{o}_{i^*}^1$ is bounded, the random variable
\begin{align}
\log\bigl(E_{i^*}^{1}(\lambda^*_{i^*})\bigr)
=
\log(1+\lambda^* \hat{o}_{i^*}^1)
\end{align}
has finite moments of all orders, implying $\sigma_s(\rho_1)<\infty$.
\section{Proof of Proposition \ref{prop:sublinear_regret}}

The regret term of the hyperparameter sequence $(\mathbf{w},\{\{\lambda_i^t\}_{t\geq1}\}_{i=1}^n)$ in \eqref{eq:time_int_regret} can be decomposed into two regret terms that separately account for the regret of weight vector $\mathbf{w}$ and the sequences of bets $\{\{\lambda_i^t\}_{t\geq1}\}_{i=1}^n$ as follows

\begin{align}
	\mathcal{R}^{j:t}&= \max_{\lambda\in\Lambda^*_i}\log E^{j:t}_{i^*}(\lambda) - \log E^{j:t}_{i^*}
	+ \log E^{j:t}_{i^*} - \log {E}^{j:t} \nonumber\\
	&= \mathcal{R}^{j:t}(\lambda^{j:t}_{i^*}) - \mathcal{R}^{j:t}(\mathbf{w}).
\end{align}
where
\begin{align}
	\mathcal{R}^{j:t}(\lambda_{i^*}^{j:t})
	&= \max_{\lambda\in\Lambda^*_i}\log E^{j:t}_{i^*}(\lambda) - \log E^{j:t}_{i^*}, \\
	\mathcal{R}^{j:t}(\mathbf{w})
	&= \log E^{j:t}_{i^*} - \log {E}^{j:t},
\end{align}
are respectively the regret of the bets sequence $\{\lambda_{i^*}^t\}_{t\geq1}$, and the regret of weight vector $\mathbf{w}$.

For any $j<t$, the regret of the  $\mathcal{R}^{j:t}(\mathbf{w})$, satisfies
\begin{align}
    \label{eq:reg_weight_vec}
	\mathcal{R}^{j:t}(\mathbf{w})
	&= \log E^{j:t}_{i^*} - \log\left(\sum^n_{i=1}w_{i}E^{j:t}_i\right)\nonumber\\
	&\leq \log E^{j:t}_{i^*} - \log w_{i^*}E^{j:t}_{i^*}\nonumber\\
	&\leq -\log w_{i^*},
\end{align}
meaning that for any interval $[j,t]$ the regret is bounded by a constant that is inversely proportional to the magnitude of the weight $w_{i^*}$ assigned to the observable $i^*$ associated with the largest expected log-growth \eqref{eq:maximum_log_growth}.

To prove that the regret of the sequence of CBCE betting parameters $\{\lambda^t_{i}\}_{t\geq1}$ has sublinear strongly adaptive regret we resort to the static regret guarantees of universal portfolio schemes \cite{cover1991universal} and on the properties of the CBCE meta-algorithm \cite[Theorem 2]{jun2017improved}.
 
First, we notice that by introducing a slackness parameter $\epsilon_{\lambda}>0$ and restricting the bet ranges $\Lambda_i\in(-u_i^{-1}+\epsilon_{\lambda},\, -l_i^{-1}-\epsilon_{\lambda})$, the log-increment loss used to evaluate experts is bounded as $\max_{\lambda\in\Lambda_i}|\log\bigl(L_{i}^{t}(\lambda)\bigr)|<\log(\epsilon_{\lambda})<\infty$ almost surely. For any observable index $i$ and any interval $\mathcal{I}=[t_1,t_2]$ with $t_2-t_1\geq 2$, the universal portfolio betting parameter sequence $\{\lambda^t_{I,i}\}^{t_2}_{t=t_1}$ in \eqref{eq:expert_bet} satisfies the static regret bound on the log-growth of the test statistic \eqref{eq:baseline_partial_prods} \cite{cover1991universal}
\begin{align}
    \max_{\lambda\in\Lambda_i}\log E^{t_1:t_2}_{i}(\lambda) - \sum^{t_2}_{t=t_1}\log(1+\lambda^t_{I,i}\hat{o}^t_i)\leq \frac{\log(t_2-t_1+1)}{2}+\log(2)
\end{align}

Leveraging the properties of the CBCE meta-algorithm, it is possible to turn the experts' static regret guarantee into a strongly adaptive regret guarantee for the sequence of CBCE betting parameters $\{\lambda^t_{i}\}_{t\geq1}$. In particular, from \cite[Theorem 2]{jun2017improved}, there exists a constant $C<\infty$ such that for $t>j+1$ it holds
\begin{align}
    \label{eq:reg_bet_params}
	\mathcal{R}^{j:t}(\lambda^{j:t}_{i^*})\leq\max_{\lambda\in \Lambda_i}\log E^{j:t}_{i^*}(\lambda) - \log E^{j:t}_{i^*}\leq C\sqrt{(t-j)(7\ln(t)+5)}
\end{align}
A strongly adaptive regret guarantee for the hyperparamaters $(\mathbf{w},\{\{\lambda^t_{i}\}_{t\geq1}\}_{i=1}^n)$ is obtained summing the regret guarantees \eqref{eq:reg_weight_vec} and \eqref{eq:reg_bet_params} as
\begin{align}
	\mathcal{R}^{j:t}\leq C\sqrt{(t-j)(7\ln(t)+5)}-\log w_{i^*}.
\end{align}
This regret certificate, combined with the upper bound \eqref{eq:finite_alpha_bound} provides the following non-asymptotic detection delay guarantee for eSCD and eSCD-CU with CBCE betting,
\begin{align}
    \label{eq:non_asympt_del_ub}
	\tau^*(\rho_1)\leq 2&+ \left\lceil\frac{(1+\delta)\log(1/\alpha)}{D^*(\rho_1)}\right\rceil+\frac{2^s \sigma_s(\rho_1)}{(s/2-1)}\left(\frac{1+\delta}{D^*(\rho_1)}+\frac{1}{\log(1/\alpha)}\right)\nonumber\\
	&+\frac{2(1+\delta)\alpha^{\delta/2}}{\delta D^*(\rho_1)}+\sum^\infty_{k=m+\nu}\mathds{1}\left\{C\sqrt{\frac{(7\ln(k)+5)}{k-\nu}}-\frac{\log w_{i^*}}{k-\nu}\geq \frac{\epsilon}{2}\right\}.
\end{align}
The non-asymptotic guarantee \eqref{eq:non_asympt_del_ub} for $\alpha\to 0^+$ and $\delta\to 1$ recovers the asymptotic guarantee \eqref{eq:optimal_detection_delay}.
\end{document}